\documentclass[11pt,english]{article}
\usepackage[T1]{fontenc}
\usepackage[latin9]{inputenc}
\usepackage{babel}
\usepackage{verbatim}
\usepackage{amsmath}
\usepackage{amsthm}
\usepackage{amssymb}
\usepackage{float}
\usepackage{graphicx}
\usepackage[unicode=true,
 bookmarks=false,
 breaklinks=false,pdfborder={0 0 1},backref=none,colorlinks=false]
 {hyperref}

\makeatletter

\usepackage{amstext}
\usepackage{enumitem}

\usepackage{thmtools}
\usepackage{thm-restate}

\usepackage{fullpage}

\usepackage{verbatim}
\usepackage{times}

\theoremstyle{plain}
\newtheorem{theorem}{Theorem}
\newtheorem{lemma}[theorem]{Lemma}
\newtheorem{obs}[theorem]{Observation}

\theoremstyle{plain}
\theoremstyle{definition}
\newtheorem{defn}[theorem]{\protect\definitionname}
\theoremstyle{plain}
\newtheorem{prop}[theorem]{\protect\propositionname}
\theoremstyle{remark}
\newtheorem{claim}[theorem]{\protect\claimname}

\ifx\proof\undefined
\newenvironment{proof}[1][\protect\proofname]{\par
    \normalfont\topsep6\p@\@plus6\p@\relax
    \trivlist
\itemindent\parindent
\item[\hskip\labelsep\scshape #1]\ignorespaces
}{%
\endtrivlist\@endpefalse
}
\providecommand{\proofname}{Proof}
\fi

\usepackage{babel}
\usepackage{authblk}
\author[1]{Waldo G\'alvez}
\author[2]{Arindam Khan} 
\author[3]{Mathieu Mari}
\author[4]{Tobias M\"{o}mke}
\author[5]{Madhusudhan Reddy}
\author[6]{Andreas Wiese}   

\affil[1]{Technical University of Munich, Germany. 
    \texttt{galvez@in.tum.de}\thanks{Supported by the European Research Council, Grant Agreement No. 691672, project APEG.}}
\affil[2]{Indian Institute of Science, Bengaluru, India.
        \texttt{arindamkhan@iisc.ac.in}}
\affil[3]{University of Warsaw, Poland.
            \texttt{mathieu.mari@ens.fr}\thanks{Supported by the ERC CoG grant TUgbOAT no 772346.}}
\affil[4]{University of Augsburg, Germany. \texttt{moemke@informatik.uni-augsburg.de}\thanks{Partially supported by the DFG Grant 439522729 (Heisenberg-Grant)}}
\affil[5]{Indian Institute of Technology,  Kharagpur, India. \texttt{pmsreddifeb18@gmail.com}}
\affil[6]{Universidad de Chile, Chile. \texttt{awiese@dii.uchile.cl}\thanks{Partially supported by FONDECYT Regular grant 1200173.}}

 \usepackage[textsize=tiny,textwidth=2cm,color=green!50!gray]{todonotes} 
 \newcommand{\awr}[1]{}
 \newcommand{\mmr}[1]{}
 \newcommand{\tmr}[1]{}
 \newcommand{\madr}[1]{}
 \newcommand{\arir}[1]{}
 
 \usepackage{xcolor}
 \newcommand{\note}[1]{{#1}}
 \newcommand{\aw}[1]{{#1}}
 \newcommand{\tm}[1]{{#1}}
 \newcommand{\mm}[1]{{#1}}
 \newcommand{\mad}[1]{{#1}}
 \newcommand{\ari}[1]{{#1}}

\makeatother

\providecommand{\claimname}{Claim}
\providecommand{\definitionname}{Definition}
\providecommand{\propositionname}{Proposition}

\begin{document}
\global\long\def\R{\mathcal{R}}%
\global\long\def\I{\mathcal{I}}%
\global\long\def\OPT{\mathrm{OPT}}%
\global\long\def\N{\mathbb{N}}%
\global\long\def\P{\mathcal{P}}%
\global\long\def\DP{\mathrm{DP}}%
\global\long\def\kone{k_{1}}%
\global\long\def\appone{30}%
\global\long\def\eps{\varepsilon}%
\title{A $(2+\varepsilon)$-Approximation Algorithm for Maximum Independent Set of Rectangles}
\maketitle
\begin{abstract}
We study the Maximum Independent Set of Rectangles (MISR) problem,
where we are given a set of axis-parallel rectangles in the plane
and the goal is to select a subset of non-overlapping rectangles of
maximum cardinality. In a recent breakthrough, Mitchell \cite{mitchell2021approximating}
obtained the first constant-factor approximation algorithm for MISR.
His algorithm achieves an approximation ratio of 10 and it is based
on a dynamic program that intuitively recursively partitions the input
plane into special polygons called \emph{corner-clipped rectangles (CCRs)}, 
without intersecting certain special horizontal line segments called
\emph{fences}.

In this paper, we present a $(2+\varepsilon)$-approximation algorithm for MISR which
is also based on a recursive partitioning scheme. First, we use
a partition into a class of axis-parallel polygons with constant complexity each that are more general than CCRs.
This allows us to provide an arguably simpler
analysis and at the same time already improves the approximation ratio to 6. Then, using a more
elaborate charging scheme and a recursive partitioning into 
general axis-parallel polygons with constant complexity,
we improve our approximation
ratio to~$2+\varepsilon$.
In particular, we construct a recursive partitioning based on more
general fences which can be sequences of up to $O(1/\varepsilon)$ line segments \mmr{shouldn't we say $O(1/\varepsilon)$ here ?}
each. This partitioning routine and our other new ideas may be useful for future work towards a PTAS for MISR.
\end{abstract}
\thispagestyle{empty}
\newpage{}

\setcounter{page}{1}

\section{Introduction}

Maximum Independent Set of Rectangles (MISR) is a fundamental and
well-studied problem in computational geometry, combinatorial optimization
and approximation algorithms. In MISR, we are given a set $\R$ of
$n$ possibly overlapping axis-parallel rectangles in the plane. We
are looking for a subset $\R^{*}\subseteq\R$ of maximum cardinality
such that the rectangles in $\R^{*}$ are pairwise disjoint. MISR
finds numerous applications in practice, e.g., in map labeling \cite{haunert2014labeling,doerschler1992rule},
data mining \cite{fukuda1996data} and resource allocation \cite{lewin2002routing}.

The problem is an important special case of the \textsc{Maximum Independent
Set }problem in graphs, which in general is NP-hard to approximate
within a factor of $n^{1-\eps}$ for any constant $\eps>0$~\cite{haastad1999clique}.
However, for MISR much better approximation ratios are possible, e.g.,
there are multiple $O(\log n)$-approximation algorithms \cite{CC2009,nielsen2000fast,khanna1998approximating}.
It had been a long-standing open problem to find an $O(1)$-approximation
algorithm for MISR. One possible approach for this is to compute an
optimal solution to the canonical LP-relaxation of MISR and round
it. This approach was used by Chalermsook and Chuzhoy in order to
obtain an $O(\log\log n)$-approximation~\cite{ChalermsookW21}.
The LP is conjectured to have an integrality gap of $O(1)$ which
is a long-standing open problem by itself, with interesting connections
to graph theory~\cite{chalermsook2011coloring,ChalermsookW21}. On
the other hand, it seems likely that MISR admits even a PTAS, given
that it admits a QPTAS due to Adamaszek and Wiese~\cite{adamaszek2013approximation},
and in particular one with a running time of only ~$n^{O((\log\log n/\eps)^{\aw{4}})}$
due to Chuzhoy and Ene~\cite{ChuzhoyE16}.

Recently, in a breakthrough result, Mitchell presented a polynomial
time 10-approximation algorithm~\cite{mitchell2021approximating}
and consequently solved the aforementioned long-standing open problem.
Instead of rounding the LP, he employs a recursive partitioning of
the plane into a special type of rectilinear polygons called corner-clipped
rectangles (CCRs). Given a CCR, he recursively subdivides it into
at most five smaller CCRs until he obtains CCRs which essentially
contain at most one rectangle from the optimal solution $\OPT$ each.
At the end, he outputs the rectangles contained in these final CCRs
plus some carefully chosen rectangles from $\OPT$ that are intersected
by these recursive cuts. With a dynamic program, he computes the recursive
partition that yields the largest number of output rectangles, which
in particular ``remembers'' in each step $O(1)$ rectangles that
were intersected by some previous cuts. In a structural theorem he
shows that there exists a set of at least $|\OPT|/10$ rectangles
which can be output by such a recursive partitioning, leading to the
approximation ratio of 10. The structural theorem is proved using
an exhaustive case analysis for defining the subdivision of a given
CCR, with sixty cases in total.

Mitchell's result yields several interesting open questions, most
notably whether one can improve the approximation ratio and whether
one can give a simpler analysis which does not rely on a large case
distinction. In this paper, we answer both questions in the affirmative.

\subsection{Our Contribution }

In this paper, we present a polynomial time $(2+\varepsilon)$-approximation
algorithm for MISR. In a first step, we construct a 6-approximation
algorithm which hence already improves the approximation ratio. It
uses an arguably simpler analysis, and also a less complicated dynamic
program. Similar to the result by Mitchell~\cite{mitchell2021approximating},
we use a recursive decomposition of the plane into a constant number
of (simple) axis-parallel polygons with constant complexity each.
However, instead of requiring the arising polygons to be CCRs, all
our polygons are horizontally 
(or vertically) convex (see Figure \ref{fig:horizontally_convex_polygon}),
i.e., if two points in the polygon lie on the same horizontal (vertical)
line, then any point between them is also contained in the polygon,
and we require the polygons to have a bounded number of edges.
Hence, this class of polygons is larger than CCRs. We prove that we
can recursively partition the input plane into polygons of this class,
such that at the end we extract at least $|\OPT|/6$ rectangles from
$\OPT$. In fact, we can compute such a partition with a simpler dynamic
program that does not need to remember any rectangle intersected
by previous cuts, but which instead just partitions the plane recursively.

\begin{figure}
\centering \includegraphics[width=0.7\textwidth]{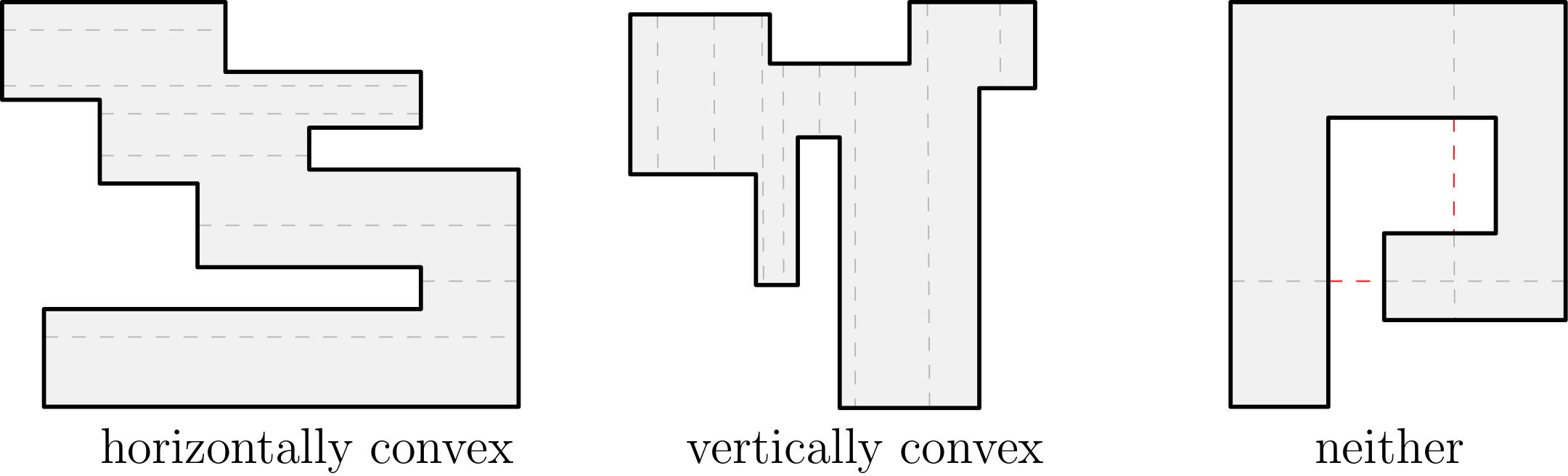}
\caption{On the left a horizontally convex polygon. On the middle, a vertically
convex polygon. On the right, an axis-parallel polygon that is neither
horizontally nor vertically convex.}
\label{fig:horizontally_convex_polygon} 
\end{figure}\awr{there was a different Figure 1, please check!}

Given a horizontally (or vertically) convex polygon, we show that
we can cut it into at most three smaller polygons (with a bounded
number of edges), such that the cut uses only one single line segment
$\ell$ that potentially intersects rectangles from $\OPT$
(and which are hence lost). In contrast, in~\cite{mitchell2021approximating}
there can be two line segments of this type which yields a higher
approximation ratio. Note that with this strategy it is unavoidable
to intersect some rectangles from $\OPT$ (unless the instance is
very easy). However, we ensure that for each rectangle $R\in\OPT$
intersected by~$\ell$, we can find \emph{two }other rectangles from
$\OPT$ that we can charge $R$ to, while in~\cite{mitchell2021approximating}
there was only one rectangle to charge to. Also, when we define our
cut, there are only two different cases that we need to consider,
which leads to an arguably simpler analysis. Based on this, we construct
a fractional charging scheme which yields an approximation ratio of~6.

We extend this approach to improve the approximation ratio even further.
In \cite{mitchell2021approximating} and in our algorithm above, \emph{fences}
are defined which are horizontal line segments, emerging from boundary
edges of the polygon. We are not allowed to cut through them in our
recursive partition, and this intuitively protect some rectangles
in $\OPT$ from being intersected (in particular the rectangles that
we charged other rectangles to). We extend this approach to more general
fences, each of them being a \emph{sequence} of $7$ axis-parallel
line segments, rather than single horizontal line segments. Like the
fences above, they protect some of the rectangles in $\OPT$ from
being deleted, since we do not allow ourselves to cut through them.
However, due to their more elaborate shape, they protect \emph{more}
rectangles in $\OPT$ and they protect them \emph{better}. For example,
we use them to ensure that each rectangle receives charge either only
from rectangles on its left, or only from rectangles on its right. 

We show that in this way we obtain an approximation ratio of 3. Note that in~\cite{mitchell2021approximating} and in
our algorithm above the rectangles in $\OPT$ are subdivided into
three groups that are denoted as \emph{horizontally nested}
rectangles, \emph{vertically nested} rectangles, and rectangles that
are \emph{neither horizontally nor vertically nested}. W.l.o.g.~we
can assume that there are at most $|\OPT|/2$ horizontally nested
rectangles, which loses a factor of 2 (in~\cite{mitchell2021approximating}
and in our argumentation). However, it turns out that our solution
still sometimes contains some horizontally nested rectangles. In this
case, we lose less than a factor of 2 in the previous step. Also,
in some settings we identify more rectangles to charge to than in
the previous argumentation. With these two ingredients we construct
a more involved fractional charging argument which improves the approximation
ratio to 3. In particular, this charging scheme crucially needs
our more general fences in order \ari{to} distribute the charges
appropriately to rectangles that are included in the computed solution.

We need that our recursive partitioning sequence does not intersect
any of our more general fences. To this end, 
we design a new recursive partitioning scheme based on axis-parallel
polygons of constant complexity, \emph{without} imposing additional
conditions on the polygons like, e.g., being horizontally convex or
a CCR. 
For our 3-approximation it would be sufficient to have such a
scheme for fences with up to 7 line segments each. However, we obtain
such a partition even for any set of $x$-monotone fences with an
arbitrarily large constant number $\tau\in\N$ of line segments each.
This can be used directly as a black-box in future work. For example,
it might be that one can design a PTAS for MISR via fences with $\tau=1/\eps^{O(1)}$
many line segments each.

\aw{Finally, we show that by using fences with $\tau=O(1/\eps)$
line segments, we improve our approximation ratio to $2+\eps$. To
this end, we replace the definition of horizontally and vertically
nested rectangles by a related but different definition of \emph{horizontally
}and \emph{vertically nice }rectangles. This simplifies the analysis
since now there are only two groups of rectangles, horizontally nice
and vertically nice rectangles. \mmr{I'm not sure this is really a simplification} Similar as above, by losing a factor
of 2 we assume that at least $|\OPT|/2$ rectangles are horizontally
nice. We construct a fractional charging argumentation in which each
intersected horizontally nice rectangle is charged either to one vertically
nice rectangle in our solution (i.e., that we assumed to be already
lost when we focused on the horizontally nice rectangles) or to $1/\eps$
horizontally nice rectangles in our solution. This yields an approximation
ratio of $2+\eps$.}

We hope that our other new ideas will lead to further progress towards
a PTAS for MISR.
\begin{restatable}{theorem}{threeapx}\label{theorem:3apx}For any $\eps>0$ there is a polynomial-time
$(2+\varepsilon)$-approximation algorithm for \aw{the Maximum Independent
Set of Rectangles problem.} 
\end{restatable}
\aw{We remark that in order to obtain a better approximation ratio
than 2, substantially new ideas seem to be needed. In the approach
by Mitchell~\cite{mitchell2021approximating} as well as in our argumentations,
the analysis loses a factor of 2 by focusing on the rectangles that
are not horizontally nested or that are horizontally nice, respectively.
This loses a factor of 2 in the approximation ratio. It seems unclear
how to avoid this.} 

In parallel and independently from our work, Mitchell recently improved
his 10-approximation algorithm to a $10/3$-approximation, and he
claims that this algorithm can be improved further to a~$(3+\eps)$-approximation
whose running time depends on $\eps$~\cite{DBLP:journals/corr/abs-2101-00326-v3}.
\awr{Rewrote this paragraph assuming that he will put his $10/3$-approximation
on the arxiv.}\awr{we should cite here the precise version of Mitchells
arxiv paper, since he might update his arxiv-paper further before
the SODA notification. }

\subsection{Other related work}

\awr{added a few author names} For simple geometric objects such
as disks, squares and fat objects, polynomial-time approximation schemes
(PTAS) are known for the corresponding setting of Independent Set~\cite{erlebach2005polynomial,ChanHarPeled2012}.
In the weighted case of MISR, 
each rectangle has an associated weight and the goal is to select
a maximum weight independent set. Recently, Chalermsook and Walczak
obtained an $O(\log\log n)$-approximation \cite{ChalermsookW21},
improving the previous $O(\log n/\log\log n)$-approximation by Chan
and Har-Peled~\cite{ChanHarPeled2012}. Furthermore, Marx~\cite{marx2005efficient}
showed that MISR is W{[}1{]}-hard, ruling out an EPTAS for the problem.
Grandoni et al. \cite{GKW19} presented a parameterized approximation
scheme for the problem. Fox and Pach \cite{fox2011computing} have
given an $n^{\eps}$-approximation for maximum independent set of
line segments. In fact, their result extends to the independent set
of intersection graphs of $k$-intersecting curves (where each pair
of curves has at most $k$ points in common).

MISR also has interesting connections with end-to-end cuts (called
guillotine cuts \cite{pach2000cutting}, also known as binary space
partitions \cite{de1997computational}). Due to its practical relevance
in cutting industry, guillotine cuts are well-studied for packing
problems, \aw{e.g.,} \cite{BansalLS05,KhanSocg21}. It has been
conjectured that, given a set of $n$ axis-parallel rectangles, $\Omega(n)$
rectangles can be separated using a sequence of guillotine cuts~\cite{AbedCCKPSW15}.
If true, this will imply an $O(n^{5})$-time simple $O(1)$-approximation
algorithm for MISR~\cite{AbedCCKPSW15,KhanR20}. 

There are many other related important geometric optimization problems,
such as Geometric Set Cover \cite{chan2012weighted,chan2014exact,mustafa2014settling},
Geometric Hitting Set \cite{chekuri2020fast,agarwal2014near,mustafa2010improved},
2-D Bin Packing \cite{bansal2010new,jansen2016new,bansal2014binpacking},
Strip Packing \cite{harren20145,JansenR19,galvez2020tight}, 2-D Knapsack
\cite{jansen2004rectangle,GalvezGHI0W17,GalSocg21}, Unsplittable
Flow on a Path \cite{chakrabarti2002approximation,bansal2006quasi,bonsma2014constant,anagnostopoulos2018mazing,GMW018},
Storage Allocation problem \cite{bar2001unified,bar2017constant,MomkeW20},
etc. We refer the readers to \cite{CKPT17} for a literature survey.\awr{TODO:
cite more papers in batch-citations}

\section{Dynamic program}

\label{sec:dynamic_program}

We assume that we are given the set $\R$ 
with 
$n$ axis-parallel rectangles in the plane such that each rectangle
$R\in\R$ 
 is specified by its two opposite corners 
$(x_l,y_b)\in\mathbb{N}^{2}$ and $(x_r,y_t)\in\mathbb{N}^{2}$, with $x_l<x_r$
and $y_b<y_t$, so that $R:=\{(x,y)\in\mathbb{R}^{2}\mid x_l<x<x_r \wedge y_b<y<y_t\}$
(i.e., the rectangles are open sets). By a standard preprocessing~\cite{AdamaszekHW19},
we can assume that, for each rectangle $R\in\R$, we have that
$x_l,x_r,y_b,y_t\in\{0,1, \dots, 2n-1\}$.
In particular, all input rectangles are contained in the square $S:=[0,2n-1]\times[0,2n-1]$.

Our algorithm is a geometric dynamic program (similar as in \cite{AdamaszekHW19,ChuzhoyE16,mitchell2021approximating})
which, intuitively, recursively subdivides $S$ into smaller polygons
until each polygon contains only one rectangle from the optimal solution
$\OPT$. For each of the latter polygons, it selects one input rectangle
that is contained in the polygon, and finally outputs the set of all
rectangles that are selected in this way. During the recursion, we
ensure that each arising polygon has only $O(1)$ edges \aw{that} are
all axis-parallel with integral coordinates. 
This ensures that there are only $n^{O(1)}$ possible
polygons of this type, which allows us to define a dynamic program
that computes the best recursive partition of $S$ in time $n^{O(1)}$.
Note that the line segments defining the recursive subdivision of
$S$ might intersect rectangles from $\OPT$ and those will not be
included in our solution.

Our dynamic program has a parameter $k\in\N$. It has a dynamic programming
table with one cell for each \aw{simple} polygon $P\subseteq S$ with at most
$k$ axis-parallel edges, such that the endpoints of each edge have
integral coordinates. 

Denote by $\P(k):=\P$ the
set of polygons corresponding to the DP-cells. For each $P\in\P$,
the dynamic program computes a solution $\DP(P)\subseteq\R$ consisting
of rectangles from $\R$ contained in $P$. For computing these solutions,
we order the polygons in $\P$ according to any partial order $\prec$
in which, for each $P,P'\in\P$ with $P\subsetneq P'$, it holds that
$P\prec P'$. We consider the polygons in $\P$ in this order so as
to compute their respective solutions $\DP(P)$. Consider a polygon
$P\in\P$. If $P$ does not contain any rectangle from $\R$ then
we define $\DP(P):=\emptyset$ and stop. Similarly, if $P$ contains
only one rectangle $R\in\R$ then we define $\DP(P):=\{R\}$ and stop.
Otherwise, 
the DP 
tries all subdivisions of $P$ into at most three 
polygons $P_{1},P_{2},P_{3}\in\P$ with at most $k$ axis-parallel edges each, looks up their corresponding (already
computed) solutions and defines their union $\DP(P_{1})\cup\DP(P_{2})\cup\DP(P_{3})$
as a \emph{candidate solution }for $P$. 
Finally, we define $\DP(P)$ to be the candidate
solution with largest cardinality. At the very end, we output $\DP(S)$. 

\begin{lemma} \label{lemma:DP-running-time}Parameterized by $k\in\N$,
the running time of the dynamic program is $O(n^{5k/2})$. \end{lemma} 

In order to analyze the DP, we introduce the concept of $k$\emph{-recursive
partitions. }Intuitively, the solution computed by the DP corresponds
to a recursive partition of $S$ into polygons in $\P$, in which
each arising polygon $P$ is further subdivided into at most three
polygons $P_{1},P_{2}$ and $P_{3}$, or instead we select \aw{at most} one rectangle
$R\subseteq P$ and do not partition $P$ further. This can be modeled
as a tree as given in the following definition. 
\begin{defn}
A \emph{$k$-recursive partition for a set }$\R'\subseteq\R$ consists
of a rooted tree with vertices $V$ such that 
\begin{itemize}[noitemsep]
\item for each node $v\in V$ there is a corresponding polygon $P_{v}\in\P(k)$, 
\item for the root node $r\in V$ \aw{it holds that} $P_{r}=S$, 
\item each internal node $v$ has at most {\em three} children $C_{v}\subseteq V$
such that $P_{v}=\dot{\cup}_{v'\in C_{v}}P_{v'}$, 
\item for each leaf $v\in V$, $P_{v}$ contains at most one rectangle in
$\R'$. 
\item for each rectangle $R'\in\R'$, there exists a leaf $v\in V$ such
that $R'\subseteq P_{v}$ and $R'\cap P_{v'}=\emptyset$ for each
leaf $v'$ with $v'\ne v$. 
\end{itemize}
\end{defn}

Note that the rectangles in $\R'$ are
pairwise disjoint due to the last property. Also, for each internal
node $v$, the polygons of its children are disjoint.

\begin{lemma} \label{lemma:k-recursive-sufficient}Given an input set
of rectangles $\R$, if there exists a $k$-recursive partition for
a set $\R'\subseteq\R$, then on input $\R$ our dynamic program computes
a solution $\tilde{\R}\subseteq\R$ with $|\tilde{\R}|\ge|\R'|$.
\end{lemma}

In the following section, we will prove the following lemma. 
\begin{lemma}\label{lemma:recursive-partition}
For an arbitrary input set of rectangles $\R$, there exists a $26$-recursive
partition for some set $\R'\subseteq\R$ with $|\R'|\ge\frac{1}{6}|\OPT|$.
\end{lemma} This yields the following theorem. 
In Appendix~\ref{sec:3apx} we will improve the approximation ratio to 3, and hence prove Theorem~\ref{theorem:3apx}.

\begin{theorem}\label{theorem:6apx} There
is a polynomial-time 6-approximation algorithm for \aw{the Maximum Independent Set of
Rectangles problem.}
\end{theorem}


\section{Recursive cutting sequence}

In this section our goal is to prove Lemma~\ref{lemma:recursive-partition}. 
Consider an optimal solution $\OPT$. 
\mm{
We will construct a $26$-recursive partition for a set $\R'\subseteq \OPT$, such that $|\R'|\ge |\OPT|/6$, and such that the axis-parallel polygons considered in this recursive partition are all horizontally convex (or all vertically convex). 
}
\begin{defn}
A polygon $P$ is \emph{horizontally (resp. vertically) convex} if,
for any two points $x,y\in P$ lying on the same horizontal (resp.
vertical) line $\ell$, the line segment connecting $x$ and $y$
is contained in $P$. 
\end{defn}
\mm{
Note that a rectangle $P$ is both horizontally and vertically convex,
and this holds in particular for $S$. }
Like Mitchell~\cite{mitchell2021approximating},
we first extend each rectangle $R\in\OPT$ in order to make it maximally
large in each dimension. Formally, we consider the rectangles in $\OPT$
in an arbitrary order. For each $R\in\OPT$, we replace $R$ by a
(possibly) larger rectangle $R'$ such that $R\subseteq R'\subseteq S$,
and if we enlarged $R'$ further by changing any one of its four coordinates,
then we would intersect some other rectangle in $\OPT$ or it would no
longer be true that $R'\subseteq S$. Denote by $\OPT'$ the resulting
solution. \begin{lemma} \label{lemma:OPT-prime-ok}For every $k\in\N$,
if there is a $k$-recursive partition for a set $\R'\subseteq\OPT'$,
then there is also a $k$-recursive partition for a set $\tilde{\R}\subseteq\OPT$
with $|\R'|=|\tilde{\R}|$. \end{lemma}

Our goal is now to prove that there always exists a $k$-recursive
partition for a subset $\R'\subseteq\OPT'$ \aw{with $|\R'|\ge |\OPT|/6$.} 
As in \cite{mitchell2021approximating}, we define \emph{nesting
relationships} for the rectangles in $\OPT'$ (see Figure~\ref{fig:nested_rectangles}).
Consider a rectangle $R\in\OPT'$. Note that each of its four edges
must intersect the edge of some other rectangle $R'\in\OPT'$ or some
edge of $S$. We say that $R$ is \emph{vertically nested }if its
top edge or its bottom edge is contained in the interior of an edge
of some other rectangle $R'\in\OPT'$ or in the interior of an edge
of $S$. Similarly, we say that $R$ is \emph{horizontally nested
}if its left edge or its right edge is contained in the interior of
an edge of some other rectangle $R'\in\OPT'$ or in the interior of
an edge of $S$.

\begin{figure}
    \begin{minipage}[t]{0.49\textwidth}
\centering \includegraphics[page=15, height=4.5cm]{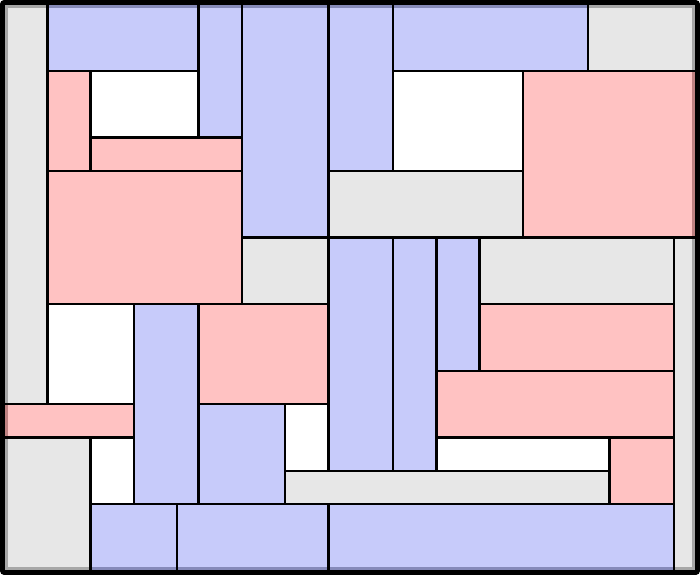}
\caption{Example of maximal rectangles $\protect\OPT'$. Red rectangles are
horizontally nested. Blue ones are vertically nested. Gray rectangles
are neither vertically nested nor horizontally. White areas are not
covered by any rectangle in $\protect\OPT'$. }
\label{fig:nested_rectangles} 
\end{minipage} \hfill
\begin{minipage}[t]{0.49\textwidth}
\centering \includegraphics[page=3, trim=100 110 185 140, clip, height=4.5cm]{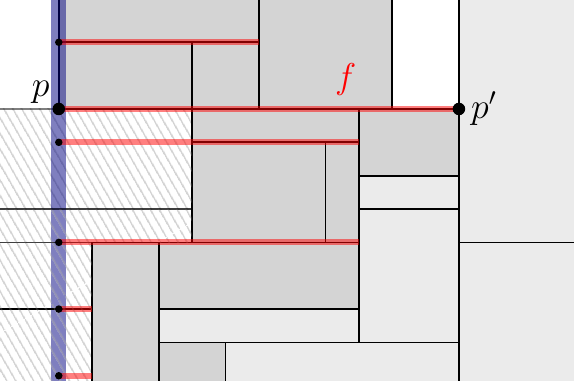}
\caption{Some line fences (red) emerging from the boundary of $P$ (blue).
Shaded rectangles are the rectangles in $\protect\OPT'$ that intersect
the boundary of $P$. 
White spaces indicate points that are not covered
by rectangles of $\protect\OPT'$ contained in $P$. 
Darker gray rectangles
correspond to rectangles in $\protect\OPT'(P)$ that are protected
by some line fence. }
\label{fig:line_fence} 
\end{minipage}
\end{figure}

\begin{prop}[\cite{mitchell2021approximating}]
\label{prop:not-both-nested}A rectangle $R\in\OPT'$ cannot be both vertically
\emph{and }horizontally nested; however, it is possible for $R$ to
be neither vertically nor horizontally nested. 
\end{prop}

We assume w.l.o.g.~that at most half of the rectangles in $\OPT'$
are horizontally nested (which will lose a factor of 2 in our
approximation ratio). Assuming this, the polygons in our recursive
partition will all be horizontally convex. 
\aw{Intuitively, we want that $\R'$} contains at least \aw{a third of} the rectangles that are not
horizontally nested \aw{which yields a factor of 6.} However, it might also contain rectangles that
are horizontally nested \aw{and that will pay for rectangles in $\OPT' \setminus \R'$ that are not horizontally nested}.  \awr{changed in many places ``nested horizontally'' $\rightarrow$ ``horizontally nested''}

\subsection{\aw{Definition of recursive partition}\label{subsec:simple-recursion}}
In order to describe \aw{our}
recursive partition, we initialize the corresponding tree $T$
with a root $r$ \aw{for which $P_r=S$}. We define now a recursive procedure that takes as input a
so far unprocessed vertex $v$ of the tree, corresponding to some
polygon $P=P_{v}$. It either partitions $P$ further (hence adding
children to $v$) or assigns \aw{at most one} rectangle $R\in\OPT'$ to $v$ and
does not add children to $v$. We denote by $\OPT'(P)\subseteq\OPT'$
the subset of rectangles of $\OPT'$ that are contained in $P$. If
$\OPT'(P)=\emptyset$, then we \aw{do not process $v$ further.} \awr{it said before that we delete $v$, however, then I believe that property 3 of Def. 1 is no longer true for the parent of $v$}
If $|\OPT'(P)|=1$,
then we add the single rectangle in $\OPT'(P)$ to $\R'$, assign
it to $v$, and do not process $v$ further. Assume now that $|\OPT'(P)|>1$.
We classify the vertical edges of $P$ as \emph{left vertical edges
}and \emph{right vertical edges }(see Figure~\ref{fig:vertical_edges}
in the appendix). 
\begin{defn}
For a vertical edge $e$ of an \mm{axis-parallel} polygon $P$, we
say that $e$ is \emph{left vertical }if its interior contains a point
$p=(p_{x},p_{y})$ such that the point $(p_{x}+\frac{1}{2},p_{y})$ is in $P$, and \emph{right vertical}, otherwise. 
\end{defn}

For each point $p$ with integral coordinates on a left vertical edge
of $P$, we define a \emph{line fence emerging from $p$}, see Figure
\ref{fig:line_fence}. If there is a point $p'\in P$ such that (i)
$p'$ and $p$ have the same $y$-coordinate, (ii) the horizontal
line segment connecting $p$ and $p'$ intersects\footnote{\aw{Recall} that the rectangles are open sets. \aw{Thus,} when a line segment $\ell$ intersects
a rectangle $R$, \aw{this means that $\ell$ contains} some point of the interior of $R$.
In particular, a line segment that (completely) contains the edge of a rectangle
$R$ does not intersect $R$.
}\awr{How about omtting the third sentence in the footnote? Is this property used somewhere?} no rectangle of $\OPT'(P)$ and (iii) $p'$ is contained in the interior
of the left side of a rectangle $R\in\OPT'(P)$, or the top right
corner or the bottom right corner of a rectangle $R\in\OPT'(P)$,
then we create a line fence that consists of the horizontal line segment
from $p$ to $p'$. Notice that if $p$ is contained in the interior
of a left edge of a rectangle in $\OPT'(P)$, then $p'=p$ and the
fence emerging from $p$ consists only of a single point. We call
$p'$ the \emph{endpoint} of the fence emerging from~$p$. We define
line fences emerging at points of right vertical edges of $P$ in
a symmetric manner. Denote by $F(P)$ the set of all fences created
in this way.

When we partition $P$, we will cut $P$ along line segments such
that (i) no horizontal line segment intersects a rectangle in $\OPT'(P)$
and (ii) no interior of a vertical line segment 
\mm{intersects the interior of a} line fence in $F(P)$.
\aw{Intuitively,} the line fences protect some rectangles in
$\OPT'(P)$ from being intersected by line segments defined in future
iterations of the partition. This motivates the following definition. 
\begin{defn}
Given a horizontally convex polygon $P$, we say that a rectangle
$R\in\OPT'(P)$ is \emph{protected in $P$} if there exists a line
fence $f\in F(P)$ such that the top edge or the bottom edge of $R$
is contained in $f$. 
\end{defn}

We will ensure that a protected rectangle will not be intersected
when we cut $P$ by means of line fences in $F(P)$. We apply the following lemma to $P$.

\begin{restatable}[Line-partitioning Lemma]{lemma}{lemmapartitionsix}
\label{lemma:partitionningLemma6} Given a horizontally convex (resp.
vertically convex) polygon $P\in\mathcal{P}(26)$, such that $P$
contains at least two rectangles from $\OPT'$, there exists a set
$C$ of line segments with integral coordinates such that: 
\begin{enumerate}[noitemsep]
\item[(1)] $C$ is composed of at most $8$ horizontal or vertical line segments
that are all contained in $P$. 
\item[(2)] $P\setminus C$ has two or three connected components, and each of
them is a horizontally (resp. vertically) convex polygon in $\mathcal{P}(26)$. 
\item[(3)] There is a vertical (resp. horizontal) line segment \tm{$\ell \in C$
    such that $\ell$} intersects all the rectangles in $\OPT'(P)$ that are intersected
    by $C$.
\item[(4)] \tm{The line segment }$\ell$ does not intersect any rectangle that is protected in $P$. 
\end{enumerate}
\end{restatable}

We introduce an (unprocessed) child vertex of $v$ corresponding to
each connected component of $P\setminus C$ which completes the processing
of $P$.

We apply the above procedure recursively to each unprocessed vertex
$v$ of the tree until there are no more unprocessed vertices. Let
$T$ denote the tree obtained at the end, and let $\R'$ denote the
set of all rectangles that we assigned to some leaf during the recursion.
One can easily see that if a rectangle $R\in\OPT'$ is protected in
some polygon $P_{v}\supseteq R$ corresponding to a node $v\in T$,
then $R$ will be protected in each polygon $P_{v'}\subseteq P_{v}$
where $v'$ is a descendant of $v$ in $T$. This implies that $R\in\R'$.%

\subsection{Analysis}

We want to prove that $|\R'|\ge|\OPT|/6$. Consider an internal
node $v$ of the tree and let $\ell_{v}$ be the corresponding line
segment $\ell$ due to Lemma~\ref{lemma:partitionningLemma6}, defined
above for partitioning $P_{v}$. We define a charging scheme for the
rectangles in $\OPT'$ that are intersected by $\ell_{v}$ and are
not horizontally nested. For any such rectangle $R$, we will identify
two rectangles $R_{L}$ and $R_{R}$ in $\OPT'(P_{v})$ such that
$R_{L}$ lies on the left of $R$ and $R_{R}$ lies on the right of
$R$, and assign a charge of $1/2$ to each of them, \aw{and} thus \aw{a} total
charge of $1$. More precisely, we will assign each of these charges
to some corner of $R_{L}$ and $R_{R}$, respectively, and ensure
that in the overall process each corner of each rectangle is charged
at most once. Thus, each rectangle receives a total charge of \ari{at most} $2$.
Furthermore, if a rectangle receives a charge (to one of its corners),
then it will be protected by the fences for the rest of the partitioning
process.

One key difference to the algorithm of Mitchell \cite{mitchell2021approximating}
is that, in our algorithm, each application of Lemma~\ref{lemma:partitionningLemma6}
yields only \emph{one} line segment $\ell$ that might intersect rectangles
from $\OPT'$. In the respective routine in \cite{mitchell2021approximating}
there can be two such line segments, and a consequence is that for
each intersected rectangle $R\in\OPT'$ there might be only one other
rectangle from $\OPT'$ to charge, rather than two. Furthermore, our
proof of Lemma~\ref{lemma:partitionningLemma6} is arguably simpler.

\paragraph{The charging scheme.}

We now explain how to distribute the charge from rectangles that are
intersected by $\ell_{v}$ (and are thus not in $\R'$). 
\begin{defn}
	We say that a rectangle $R\in\OPT'$ \emph{sees the top-left corner $c$} of
	a rectangle $R'\in\OPT'$ \emph{on its right} if there is a horizontal line segment $h$ that connects a point $p$ on the right edge of $R$  with $c$, such that $h$ does not intersect any rectangle in $\OPT'$, 
	\aw{$p$ is not the bottom-right corner of $R$},
	\awr{this was written inside parenthesis before. IMO parts of the definition should not be inside parenthesis.}
 	and $h$ does not contain \aw{the top} edge of \aw{any} other rectangle in $\aw{R''\in} \OPT'$.
	
	\label{def:seeing} 
\end{defn}

\begin{figure}
\centering \includegraphics[page=11, trim=170 170 165 60, clip, width=0.4\textwidth]{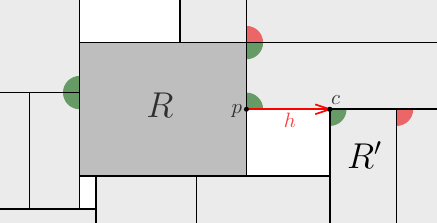}
\caption{Green marks indicate all the corners of the rectangles in $\protect\OPT'$
that are seen by the rectangle $R$. Corners displayed in red are examples of corners not seen by $R$.  The one on the top is a bottom-left corner and here $p$ would be the top-right corner of $R$, which is excluded. The red corner on the right is not seen by $R$ because here $h$ would contain the top edge of $R'$. }
\label{fig:def_see}
\end{figure}
\awr{Can we add in the figure some examples for corners that $R$ does NOT see (e.g., in red)?}

\aw{The last two conditions in the definition ensure that $R'$ is not completely below $R$ (i.e., below the line that contains the bottom edge of $R$), and that $c'$ is not ``behind'' $R''$.}
We define the bottom-left corner \emph{seen by} $R\in\OPT'$ \emph{on its right} as well as the corners \emph{seen by} $R\in\OPT'$ \emph{on its left, top } and \emph{ bottom}
in a symmetric manner. See Figure \ref{fig:def_see}.

It is easy to see that if a rectangle $R$ is horizontally nested,
then there is at least one side (left \aw{or} right) on which $R$ does
not see any corner. On the other hand, if $R$ is not horizontally
nested, then on its left it sees at least one corner of a rectangle
in $\OPT'(P)$, and similarly on its right. Intuitively, we will later
charge $R$ to these rectangles in $\OPT'(P)$. 

\begin{lemma} \label{lemma:charging_process}Let $P$ be an axis-parallel polygon, and let $R$ be a rectangle in $\OPT'(P)$ that is
not protected in $P$ and not horizontally nested. Then, $R$ sees
at least one corner of a rectangle in $\OPT'(P)$ on its left, and
at least one corner of another rectangle in $\OPT'(P)$ on its right.
\end{lemma}
For every node $v\in T$ and every rectangle $R\in\OPT'(P_{v})$ \aw{that is} not
horizontally nested \aw{and} that is intersected by $\ell_{v}$, we assign
a (fractional) charge of $1/2$ to a corner of a rectangle in $\OPT'(P_{v})$
that $R$ sees on its left, and a charge of $1/2$ to a corner of
a rectangle in $\OPT'(P_{v})$ that $R$ sees on its right.

We prove now that if some rectangle $R\in\OPT'$ is charged at some
point, then $R\in\R'$. The reason is that when $R$ is charged due
to a vertical line segment $\ell_{v}$, then in the subsequent subproblems
(i.e., corresponding to the children of $v$) $R$ will be protected.

\begin{lemma} If a rectangle $R'\in\OPT'$ receives a charge to at
least one of its corner, then $R'\in\R'$. \label{lemma:charged_implies_protected}
\end{lemma}

In the next lemma, we show that each corner of a rectangle $R'\in\R'$
can be charged at most once. Hence, each rectangle receives a total
fractional charge of at most 2. 

\begin{lemma} Each corner of a rectangle in $\R'$ is charged at
most once. \label{lemma:corner_charged_once} \end{lemma}

As a consequence, each rectangle in $\R'$ needs to pay for at most two other
rectangles that are not horizontally nested and that were intersected,
which loses a factor 3. We lose another factor 2 since we assumed
that at most half of the rectangles in $\OPT'$ are not horizontally nested.
This yields a factor of 6 overall. 

\begin{lemma} \label{lemma:factor-6}We have $|\R'|\ge|\OPT|/6$. \end{lemma}

\subsubsection{Proof of the Line-partitioning Lemma}

Now we prove Lemma~\ref{lemma:partitionningLemma6}. We assume w.l.o.g.~that
$P$ is horizontally convex with at most $26$ edges. We denote by $v\le 13$ the number of vertical edges of $P$.  We denote by $L(P)$ and $R(P)$ the
set of the left and right vertical edges of $P$, respectively. Assume
w.l.o.g.~that $|L(P)|\ge|R(P)|$. Let $e_{1},...,e_{s}$ denote the
left vertical edges of $P$, ordered from top to bottom. We have $s\ge v/2$.  
Consider
the edges $E_{M}:=\{e_{\left\lfloor s/3\right\rfloor +1},...,e_{\left\lceil 2s/3\right\rceil }\}$,
which are essentially the edges in the middle third of $e_{1},...,e_{s}$,
see Figure \ref{fig:cuts_line}. Let $f\in F(P)$ be a line fence
in $F(P)$ emerging from a point $p$ on an edge in $E_{M}$ such
that, among all such fences, its endpoint $p'$ is the furthest to
the right. Imagine that we define a vertical ray that emerges in $p'$ and
that is oriented downward. We follow this ray until we reach
a point $q'_{b}$ such that (i) $q'_{b}$ is contained in the interior
of the top edge of some rectangle $R_{b}\in\OPT'(P)$ that is protected
by some fence $g_{b}\in F(P)$, or (ii) $q'_{b}$ is contained in
some fence $g_{b}\in F(P)$ such that $q'_{b}$ is neither the first
nor the last point of $g_{b}$ (hence, intuitively $q'_{b}$ is in
the interior of $g_{b}$), or (iii) if we continued further we would
leave $P$. In the first two cases, we define $q_{b}$ to be the point
that $g_{b}$ emerges from, in the latter case we simply define $q_{b}:=q'_{b}$.
In a symmetric manner, we define a vertical ray that is oriented upward,
emerging from $p'$, and we define corresponding points $q_{t},q'_{t}$,
and possibly a corresponding fence $g_{t}\in F(P)$ and possibly a
corresponding rectangle $R_{t}\in\OPT'(P)$.

We define $C_{b}$ to be a sequence of line segments that connect
$p$ with $q_{b}$, using only points in $f$, $\overline{p'q'_{b}}$,
the top edge of $R_{b}$, the \ari{left} edge of $R_{b}$ and $g_{b}$, see Figure
\ref{fig:cuts_line}.
\arir{I changed to left edge to be consistent with the figure. Actually both left or right edge will work. May be we can just mention boundary of $R_b$ instead of top and left/right edges.}
We define $C_{t}$ similarly. Then we define
our cut $C$ by $C:=C_{b}\cup C_{t}$ with $\ell:=\overline{q_{t}'q'_{b}}$.
Note that $C$ has at most 8 edges. By construction, $\ell$ is the
only line segment in $C$ that can intersect rectangles in $\OPT'$,
however, $\ell$ does not intersect any protected rectangle. Also,
the length of $\ell$ is strictly larger than zero.

We need to argue that each connected component of $P\setminus C$
has at most 26 edges, i.e., at most $13$ vertical edges. First, $P\setminus C$ consists of at most three connected components: \aw{apart from the boundary of $P$,} the first component $P_1$ is \aw{enclosed} by $C_t$, the second component $P_2$ is \aw{enclosed} by $C_b$ and the last component $P_3$ is \aw{enclosed} by the sequence of line segments $C_r\subseteq C$ that connects $q_t$ and $q_b$. Now observe that the boundary of \ari{$P_1$} is disjoint from the edges \ari{$E_{B}:=\{e_{\left\lceil 2s/3\right\rceil +1}, \dots, e_{s}\}$}, the boundary of \ari{$P_2$} is disjoint from $E_{T}:=\{e_{1}, \dots, e_{\left\lfloor s/3\right\rfloor }\}$, and the boundary of \ari{$P_3$} is disjoint from $E_M$. 
Since $C_b$ (and $C_t$) has at least $\lfloor s/3\rfloor\ge \lfloor v/6\rfloor$ vertical edges, and $C_b$ (and $C_t$) has at most $2$ vertical edges,  the number of vertical edges of $P_1$ (and $P_2$) is at most  $v-\lfloor v/6\rfloor +2= \lceil \frac{5}{6}v\rceil +2\le\lceil \frac{5}{6} \cdot 13\rceil +2= 13$. 
Similarly,  since \ari{$E_M$} has at least $\lceil s/3\rceil\ge \lceil v/6\rceil$ vertical edges and $C_r$ has at most $3$ vertical edges,  the number of vertical edges of $P_3$ is at most  $v-\lceil v/6\rceil +3= \lfloor \frac{5}{6}v\rfloor+3\le\lfloor \frac{5}{6} \cdot 13\rfloor+3= 13$. 


Finally, we would like $P\setminus C$ to have at least two connected
components. This is clearly true if $\ell$ or $\overline{pp'}$ contain
a point in the interior of $P$. One can show that if this is not
the case, then $\overline{pp'}$ must be identical with the top or
the bottom edge of $P$ and $p'$ must be the top-right corner of
a rectangle $R\in\OPT'(P)$ (see Figure~\ref{fig:cuts_line}), we refer to Appendix~\ref{apx:details-Lemma6} for details. Due to
our choice of $f$ this implies that $\left\lfloor s/3\right\rfloor +1=1$
and hence $P$ has at most $4s\le8$ edges. Thus, we can
for example define a cut $C'$ that consists of the horizontal line
segment between $p$ and the top-left corner of $R$, the left edge
of $R$, and the bottom edge of $R$.


\begin{figure}
\centering \includegraphics[page=21,width=0.75\textwidth]{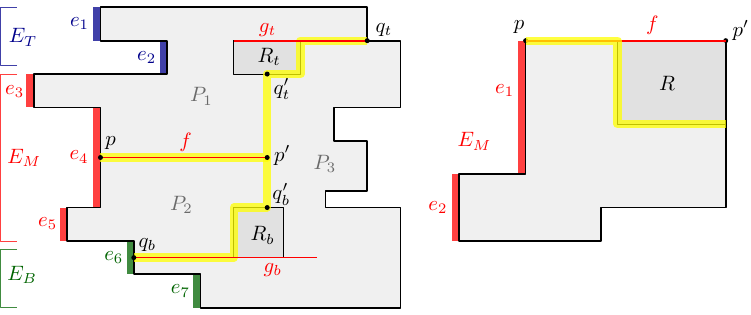}
\caption{Partitioning an horizontally convex polygon with $26$ edges (left). 
\mm{If the partitioning process fails to create two connected components, the polygon must have at most 8 edges (right).}
\aw{In both cases, the set of line segments $C$ used to cut the polygon is shown in yellow. 
}}
\label{fig:cuts_line}
\end{figure}

\section{Improving the approximation ratio to 3}

In this section, we give an overview of our additional ideas to obtain
an approximation ratio of 3; we refer to Appendix~\ref{sec:3apx}
for details. One key idea is to use more elaborate fences that are
no longer just horizontal line segments, but instead $x$-monotone
sequences of $\tau=O(1)$ line segments, see Figure~\ref{fig:5fences} in Appendix.
Like before, each of them emerges on a point of some vertical edges
of the corresponding polygon $P$. We refer to these new fences as
$\tau$\emph{-fences }and hence our previous fences are 1-fences.
In particular, $\tau$-fences protect more rectangles from $\OPT$
(if $\tau$ is sufficiently large), see Figure~\ref{fig:5fences}.

One benefit of these larger fences is the following. Suppose that
we apply a cut $C$ to a polygon $P$, let $\ell$ denote the single
(vertical) line segment in $C$ that intersects rectangles from $\OPT'$.
Assume that due to $C$, the top-right or bottom-right corner of some
rectangle $R$ is charged. In the argumentation in Section~\ref{subsec:simple-recursion},
it could happen that later the top-left or bottom-left corner of $R$
is charged by some rectangle $R'$. However, if we use $\tau$-fences
for some $\tau\ge3$, then after applying the cut $C$, such a rectangle
$R'$ is protected by a $\tau$-fence that emerges from \aw{a point
on} $\ell$ (which is a vertical edge in the connected component
of $P\setminus C$ that contains $R$). Therefore, at most two corners
from each rectangle receive a charge: either only its left or only
its right corners. Therefore, each rectangle $R\in\R'$ receives a
fractional charge of at most 1. Thus, due to this we lose only a factor
of 2, in addition to the factor of 2 that we lost by assuming that
there are at most $|\OPT|/2$ horizontally nested rectangles (and
we assumed that we lost them completely). Hence, already this improves
the approximation ratio to 4. 

\awr{Rewrote this paragraph completely AGAIN, please check} Then, we improve the approximation
ratio to 3 with the following argumentation. Suppose that a rectangle
$R\in\OPT'$ is intersected by some cut $C$ such that $R$ is not
horizontally nested. If $R$ sees two corners of (one or two) other
rectangles in $\OPT'$ on its right, then we can charge $R$ to these
two corners, instead of charging it only to one corner like in the
argumentation above. If we can do this for all intersected rectangles
$R\in\OPT'$ (towards their respective left and right), then one can show
that this already improves the approximation ratio to 3. However,
it might be that $R$ sees only one such corner $c$ on its right. 
Then we show that one of the following two cases applies. The first
case is that we identify two corners $c',c''$ belonging to rectangles
$R',R''$ on the right of $R$ such that $R$ sees $c'$ but $R$
does \emph{not} see $c''$. \aw{However, we show that if $\tau \ge 7$ then this ensures} 
that $R''\in\R'$, that $c''$ will not be charged again \mad{if $R''$ is not horizontally nested and at most twice \aw{in total} otherwise}, \note{and that the corners of $R''$ on the opposite side will never be charged.} Hence,
we can assign a charge of $1/4$ each to $c'$ and $c''$. The second
case is that $R$ sees a corner $\hat{c}$ of a horizontally nested
rectangle $\hat{R}$ on the right of $R$. In this case, we assign
a charge of $1/2$ to $\hat{c}$.  One may wonder why we can afford
to assign a charge of $1/2$ to $\hat{c}$ \aw{(and similarly before two charges of $1/4$ each to $c''$ if $R''$ is horizontally nested)}. The intuitive reason is
that we had already given up on the horizontally nested rectangles,
assuming a loss of a factor of 2. However, after charging $\hat{c}$
we ensure that $\hat{R}\in\R'$ \aw{(we show that $\tau \ge 3$ suffices)} while we had assumed that we had already
lost $\hat{R}$. Hence, we can assign a charge of 1 unit to $\hat{R}$
``for free''. 

\aw{We choose $\tau := 7$ and hence we} need a recursive partitioning scheme
that does not cut through any $\aw{7}$-fence. In
fact, we prove even a stronger statement that could be useful for
future work: we show that for any constant $\tau\in\N$ there is a
recursive partitioning scheme into polygons in $\P(30\tau+18)$ such
that each recursive cut does not cut through any $\tau$-fence. In
particular, our routine for cutting one given polygon is a generalization
of Lemma~\ref{lemma:partitionningLemma6}. The resulting polygons
might no longer be horizontally convex; we allow them to be arbitrary
simple axis-parallel polygons with $O(\tau)$ vertices and integral
coordinates. 
However, we can still establish the necessary partitioning scheme
via distinguishing only a few different cases.

\threeapx*

\appendix
\bibliographystyle{plainurl}
\bibliography{bibdb}

\begin{thebibliography}{10}

\bibitem{AbedCCKPSW15}
Fidaa Abed, Parinya Chalermsook, Jos{\'{e}}~R. Correa, Andreas Karrenbauer,
  Pablo P{\'{e}}rez{-}Lantero, Jos{\'{e}}~A. Soto, and Andreas Wiese.
\newblock On guillotine cutting sequences.
\newblock In {\em Approximation, Randomization, and Combinatorial Optimization.
  Algorithms and Techniques ({APPROX/RANDOM})}, volume~40, pages 1--19. Schloss
  Dagstuhl - Leibniz-Zentrum f{\"{u}}r Informatik, 2015.
\newblock \href {http://dx.doi.org/10.4230/LIPIcs.APPROX-RANDOM.2015.1}
  {\path{doi:10.4230/LIPIcs.APPROX-RANDOM.2015.1}}.

\bibitem{AdamaszekHW19}
Anna Adamaszek, Sariel Har{-}Peled, and Andreas Wiese.
\newblock Approximation schemes for independent set and sparse subsets of
  polygons.
\newblock {\em J. {ACM}}, 66(4):29:1--29:40, 2019.
\newblock \href {http://dx.doi.org/10.1145/3326122}
  {\path{doi:10.1145/3326122}}.

\bibitem{adamaszek2013approximation}
Anna Adamaszek and Andreas Wiese.
\newblock Approximation schemes for maximum weight independent set of
  rectangles.
\newblock In {\em 54th Annual {IEEE} Symposium on Foundations of Computer
  Science ({FOCS})}, pages 400--409. {IEEE} Computer Society, 2013.
\newblock \href {http://dx.doi.org/10.1109/FOCS.2013.50}
  {\path{doi:10.1109/FOCS.2013.50}}.

\bibitem{agarwal2014near}
Pankaj~K. Agarwal and Jiangwei Pan.
\newblock Near-linear algorithms for geometric hitting sets and set covers.
\newblock {\em Discret. Comput. Geom.}, 63(2):460--482, 2020.
\newblock \href {http://dx.doi.org/10.1007/s00454-019-00099-6}
  {\path{doi:10.1007/s00454-019-00099-6}}.

\bibitem{anagnostopoulos2018mazing}
Aris Anagnostopoulos, Fabrizio Grandoni, Stefano Leonardi, and Andreas Wiese.
\newblock A mazing 2+$\varepsilon$ approximation for unsplittable flow on a
  path.
\newblock {\em {ACM} Trans. Algorithms}, 14(4):55:1--55:23, 2018.
\newblock \href {http://dx.doi.org/10.1145/3242769}
  {\path{doi:10.1145/3242769}}.

\bibitem{bansal2010new}
Nikhil Bansal, Alberto Caprara, and Maxim Sviridenko.
\newblock A new approximation method for set covering problems, with
  applications to multidimensional bin packing.
\newblock {\em {SIAM} J. Comput.}, 39(4):1256--1278, 2009.
\newblock \href {http://dx.doi.org/10.1137/080736831}
  {\path{doi:10.1137/080736831}}.

\bibitem{bansal2006quasi}
Nikhil Bansal, Amit Chakrabarti, Amir Epstein, and Baruch Schieber.
\newblock A quasi-ptas for unsplittable flow on line graphs.
\newblock In {\em Proceedings of the 38th Annual {ACM} Symposium on Theory of
  Computing ({STOC})}, pages 721--729. {ACM}, 2006.
\newblock \href {http://dx.doi.org/10.1145/1132516.1132617}
  {\path{doi:10.1145/1132516.1132617}}.

\bibitem{bansal2014binpacking}
Nikhil Bansal and Arindam Khan.
\newblock Improved approximation algorithm for two-dimensional bin packing.
\newblock In {\em Proceedings of the Twenty-Fifth Annual {ACM-SIAM} Symposium
  on Discrete Algorithms ({SODA})}, pages 13--25. {SIAM}, 2014.
\newblock \href {http://dx.doi.org/10.1137/1.9781611973402.2}
  {\path{doi:10.1137/1.9781611973402.2}}.

\bibitem{BansalLS05}
Nikhil Bansal, Andrea Lodi, and Maxim Sviridenko.
\newblock A tale of two dimensional bin packing.
\newblock In {\em 46th Annual {IEEE} Symposium on Foundations of Computer
  Science ({FOCS})}, pages 657--666. {IEEE} Computer Society, 2005.
\newblock \href {http://dx.doi.org/10.1109/SFCS.2005.10}
  {\path{doi:10.1109/SFCS.2005.10}}.

\bibitem{bar2001unified}
Amotz Bar{-}Noy, Reuven Bar{-}Yehuda, Ari Freund, Joseph Naor, and Baruch
  Schieber.
\newblock A unified approach to approximating resource allocation and
  scheduling.
\newblock {\em J. {ACM}}, 48(5):1069--1090, 2001.
\newblock \href {http://dx.doi.org/10.1145/502102.502107}
  {\path{doi:10.1145/502102.502107}}.

\bibitem{bar2017constant}
Reuven Bar{-}Yehuda, Michael Beder, and Dror Rawitz.
\newblock A constant factor approximation algorithm for the storage allocation
  problem.
\newblock {\em Algorithmica}, 77(4):1105--1127, 2017.
\newblock \href {http://dx.doi.org/10.1007/s00453-016-0137-8}
  {\path{doi:10.1007/s00453-016-0137-8}}.

\bibitem{bonsma2014constant}
Paul~S. Bonsma, Jens Schulz, and Andreas Wiese.
\newblock A constant-factor approximation algorithm for unsplittable flow on
  paths.
\newblock {\em {SIAM} J. Comput.}, 43(2):767--799, 2014.
\newblock \href {http://dx.doi.org/10.1137/120868360}
  {\path{doi:10.1137/120868360}}.

\bibitem{chakrabarti2002approximation}
Amit Chakrabarti, Chandra Chekuri, Anupam Gupta, and Amit Kumar.
\newblock Approximation algorithms for the unsplittable flow problem.
\newblock {\em Algorithmica}, 47(1):53--78, 2007.
\newblock \href {http://dx.doi.org/10.1007/s00453-006-1210-5}
  {\path{doi:10.1007/s00453-006-1210-5}}.

\bibitem{chalermsook2011coloring}
Parinya Chalermsook.
\newblock Coloring and maximum independent set of rectangles.
\newblock In {\em 14th International Workshop on Approximation, Randomization,
  and Combinatorial Optimization ({APPROX/RANDOM})}, volume 6845, pages
  123--134. Springer, 2011.
\newblock \href {http://dx.doi.org/10.1007/978-3-642-22935-0\_11}
  {\path{doi:10.1007/978-3-642-22935-0\_11}}.

\bibitem{CC2009}
Parinya Chalermsook and Julia Chuzhoy.
\newblock Maximum independent set of rectangles.
\newblock In {\em Proceedings of the Twentieth Annual {ACM-SIAM} Symposium on
  Discrete Algorithms ({SODA})}, pages 892--901. {SIAM}, 2009.
\newblock URL: \url{http://dl.acm.org/citation.cfm?id=1496770.1496867}.

\bibitem{ChalermsookW21}
Parinya Chalermsook and Bartosz Walczak.
\newblock Coloring and maximum weight independent set of rectangles.
\newblock In {\em Proceedings of the 2021 {ACM-SIAM} Symposium on Discrete
  Algorithms ({SODA})}, pages 860--868. {SIAM}, 2021.
\newblock \href {http://dx.doi.org/10.1137/1.9781611976465.54}
  {\path{doi:10.1137/1.9781611976465.54}}.

\bibitem{chan2014exact}
Timothy~M. Chan and Elyot Grant.
\newblock Exact algorithms and apx-hardness results for geometric packing and
  covering problems.
\newblock {\em Comput. Geom.}, 47(2):112--124, 2014.
\newblock \href {http://dx.doi.org/10.1016/j.comgeo.2012.04.001}
  {\path{doi:10.1016/j.comgeo.2012.04.001}}.

\bibitem{chan2012weighted}
Timothy~M. Chan, Elyot Grant, Jochen K{\"{o}}nemann, and Malcolm Sharpe.
\newblock Weighted capacitated, priority, and geometric set cover via improved
  quasi-uniform sampling.
\newblock In {\em Proceedings of the Twenty-Third Annual {ACM-SIAM} Symposium
  on Discrete Algorithms ({SODA})}, pages 1576--1585. {SIAM}, 2012.
\newblock \href {http://dx.doi.org/10.1137/1.9781611973099.125}
  {\path{doi:10.1137/1.9781611973099.125}}.

\bibitem{ChanHarPeled2012}
Timothy~M. Chan and Sariel Har{-}Peled.
\newblock Approximation algorithms for maximum independent set of pseudo-disks.
\newblock {\em Discret. Comput. Geom.}, 48(2):373--392, 2012.
\newblock \href {http://dx.doi.org/10.1007/s00454-012-9417-5}
  {\path{doi:10.1007/s00454-012-9417-5}}.

\bibitem{Chazelle82}
Bernard Chazelle.
\newblock A theorem on polygon cutting with applications.
\newblock In {\em 23rd Annual Symposium on Foundations of Computer Science,
  Chicago, Illinois, USA, 3-5 November 1982}, pages 339--349. {IEEE} Computer
  Society, 1982.
\newblock URL: \url{https://doi.org/10.1109/SFCS.1982.58}, \href
  {http://dx.doi.org/10.1109/SFCS.1982.58} {\path{doi:10.1109/SFCS.1982.58}}.

\bibitem{chekuri2020fast}
Chandra Chekuri, Sariel Har{-}Peled, and Kent Quanrud.
\newblock Fast lp-based approximations for geometric packing and covering
  problems.
\newblock In {\em Proceedings of the 2020 {ACM-SIAM} Symposium on Discrete
  Algorithms ({SODA})}, pages 1019--1038. {SIAM}, 2020.
\newblock \href {http://dx.doi.org/10.1137/1.9781611975994.62}
  {\path{doi:10.1137/1.9781611975994.62}}.

\bibitem{CKPT17}
Henrik~I. Christensen, Arindam Khan, Sebastian Pokutta, and Prasad Tetali.
\newblock Approximation and online algorithms for multidimensional bin packing:
  {A} survey.
\newblock {\em Comput. Sci. Rev.}, 24:63--79, 2017.
\newblock \href {http://dx.doi.org/10.1016/j.cosrev.2016.12.001}
  {\path{doi:10.1016/j.cosrev.2016.12.001}}.

\bibitem{ChuzhoyE16}
Julia Chuzhoy and Alina Ene.
\newblock On approximating maximum independent set of rectangles.
\newblock In {\em {IEEE} 57th Annual Symposium on Foundations of Computer
  Science ({FOCS})}, pages 820--829. {IEEE} Computer Society, 2016.
\newblock \href {http://dx.doi.org/10.1109/FOCS.2016.92}
  {\path{doi:10.1109/FOCS.2016.92}}.

\bibitem{de1997computational}
Mark de~Berg, Otfried Cheong, Marc~J. van Kreveld, and Mark~H. Overmars.
\newblock Computational geometry: algorithms and applications, 3rd edition.
\newblock Springer, 2008.
\newblock URL: \url{https://www.worldcat.org/oclc/227584184}.

\bibitem{doerschler1992rule}
Jeffrey~S. Doerschler and Herbert Freeman.
\newblock A rule-based system for dense-map name placement.
\newblock {\em Commun. {ACM}}, 35(1):68--79, 1992.
\newblock \href {http://dx.doi.org/10.1145/129617.129620}
  {\path{doi:10.1145/129617.129620}}.

\bibitem{erlebach2005polynomial}
Thomas Erlebach, Klaus Jansen, and Eike Seidel.
\newblock Polynomial-time approximation schemes for geometric intersection
  graphs.
\newblock {\em {SIAM} J. Comput.}, 34(6):1302--1323, 2005.
\newblock \href {http://dx.doi.org/10.1137/S0097539702402676}
  {\path{doi:10.1137/S0097539702402676}}.

\bibitem{fox2011computing}
Jacob Fox and J{\'{a}}nos Pach.
\newblock Computing the independence number of intersection graphs.
\newblock In {\em Proceedings of the Twenty-Second Annual {ACM-SIAM} Symposium
  on Discrete Algorithms ({SODA})}, pages 1161--1165. {SIAM}, 2011.
\newblock \href {http://dx.doi.org/10.1137/1.9781611973082.87}
  {\path{doi:10.1137/1.9781611973082.87}}.

\bibitem{fukuda1996data}
Takeshi Fukuda, Yasuhiko Morimoto, Shinichi Morishita, and Takeshi Tokuyama.
\newblock Data mining using two-dimensional optimized accociation rules:
  Scheme, algorithms, and visualization.
\newblock pages 13--23, 1996.
\newblock \href {http://dx.doi.org/10.1145/233269.233313}
  {\path{doi:10.1145/233269.233313}}.

\bibitem{galvez2020tight}
Waldo G{\'{a}}lvez, Fabrizio Grandoni, Afrouz~Jabal Ameli, Klaus Jansen,
  Arindam Khan, and Malin Rau.
\newblock A tight (3/2+{\(\epsilon\)}) approximation for skewed strip packing.
\newblock In {\em Approximation, Randomization, and Combinatorial Optimization.
  Algorithms and Techniques ({APPROX/RANDOM})}, volume 176, pages 44:1--44:18.
  Schloss Dagstuhl - Leibniz-Zentrum f{\"{u}}r Informatik, 2020.
\newblock \href {http://dx.doi.org/10.4230/LIPIcs.APPROX/RANDOM.2020.44}
  {\path{doi:10.4230/LIPIcs.APPROX/RANDOM.2020.44}}.

\bibitem{GalvezGHI0W17}
Waldo G{\'{a}}lvez, Fabrizio Grandoni, Sandy Heydrich, Salvatore Ingala,
  Arindam Khan, and Andreas Wiese.
\newblock Approximating geometric knapsack via l-packings.
\newblock In {\em 58th {IEEE} Annual Symposium on Foundations of Computer
  Science ({FOCS})}, pages 260--271. {IEEE} Computer Society, 2017.
\newblock \href {http://dx.doi.org/10.1109/FOCS.2017.32}
  {\path{doi:10.1109/FOCS.2017.32}}.

\bibitem{GalSocg21}
Waldo G{\'{a}}lvez, Fabrizio Grandoni, Arindam Khan, Diego
  Ram{\'{\i}}rez{-}Romero, and Andreas Wiese.
\newblock Improved approximation algorithms for 2-dimensional knapsack: Packing
  into multiple l-shapes, spirals, and more.
\newblock In {\em 37th International Symposium on Computational Geometry
  ({SoCG})}, volume 189, pages 39:1--39:17. Schloss Dagstuhl - Leibniz-Zentrum
  f{\"{u}}r Informatik, 2021.
\newblock \href {http://dx.doi.org/10.4230/LIPIcs.SoCG.2021.39}
  {\path{doi:10.4230/LIPIcs.SoCG.2021.39}}.

\bibitem{GKW19}
Fabrizio Grandoni, Stefan Kratsch, and Andreas Wiese.
\newblock Parameterized approximation schemes for independent set of rectangles
  and geometric knapsack.
\newblock In {\em 27th Annual European Symposium on Algorithms ({ESA})}, volume
  144, pages 53:1--53:16. Schloss Dagstuhl - Leibniz-Zentrum f{\"{u}}r
  Informatik, 2019.
\newblock \href {http://dx.doi.org/10.4230/LIPIcs.ESA.2019.53}
  {\path{doi:10.4230/LIPIcs.ESA.2019.53}}.

\bibitem{GMW018}
Fabrizio Grandoni, Tobias M{\"{o}}mke, Andreas Wiese, and Hang Zhou.
\newblock A {(5/3} + {\(\epsilon\)})-approximation for unsplittable flow on a
  path: placing small tasks into boxes.
\newblock In {\em Proceedings of the 50th Annual {ACM} {SIGACT} Symposium on
  Theory of Computing ({STOC})}, pages 607--619. {ACM}, 2018.
\newblock \href {http://dx.doi.org/10.1145/3188745.3188894}
  {\path{doi:10.1145/3188745.3188894}}.

\bibitem{harren20145}
Rolf Harren, Klaus Jansen, Lars Pr{\"{a}}del, and Rob van Stee.
\newblock A $(5/3 + \varepsilon)$-approximation for strip packing.
\newblock {\em Comput. Geom.}, 47(2):248--267, 2014.
\newblock \href {http://dx.doi.org/10.1016/j.comgeo.2013.08.008}
  {\path{doi:10.1016/j.comgeo.2013.08.008}}.

\bibitem{haastad1999clique}
Johan H{\aa}stad.
\newblock Clique is hard to approximate within $n^{1-\varepsilon}$.
\newblock {\em Acta Mathematica}, 182(1):105--142, 1999.
\newblock \href {http://dx.doi.org/10.1007/BF02392825}
  {\path{doi:10.1007/BF02392825}}.

\bibitem{haunert2014labeling}
Jan{-}Henrik Haunert and Tobias Hermes.
\newblock Labeling circular focus regions based on a tractable case of maximum
  weight independent set of rectangles.
\newblock In {\em Proceedings of the 2nd {ACM} International Workshop on
  Interacting with Maps, MapInteract ({SIGSPATIAL})}, pages 15--21. {ACM},
  2014.
\newblock \href {http://dx.doi.org/10.1145/2677068.2677069}
  {\path{doi:10.1145/2677068.2677069}}.

\bibitem{jansen2016new}
Klaus Jansen and Lars Pr{\"{a}}del.
\newblock New approximability results for two-dimensional bin packing.
\newblock {\em Algorithmica}, 74(1):208--269, 2016.
\newblock \href {http://dx.doi.org/10.1007/s00453-014-9943-z}
  {\path{doi:10.1007/s00453-014-9943-z}}.

\bibitem{JansenR19}
Klaus Jansen and Malin Rau.
\newblock Closing the gap for pseudo-polynomial strip packing.
\newblock In {\em 27th Annual European Symposium on Algorithms ({ESA})}, volume
  144, pages 62:1--62:14. Schloss Dagstuhl - Leibniz-Zentrum f{\"{u}}r
  Informatik, 2019.
\newblock \href {http://dx.doi.org/10.4230/LIPIcs.ESA.2019.62}
  {\path{doi:10.4230/LIPIcs.ESA.2019.62}}.

\bibitem{jansen2004rectangle}
Klaus Jansen and Guochuan Zhang.
\newblock On rectangle packing: maximizing benefits.
\newblock In J.~Ian Munro, editor, {\em Proceedings of the Fifteenth Annual
  {ACM-SIAM} Symposium on Discrete Algorithms ({SODA})}, pages 204--213.
  {SIAM}, 2004.
\newblock URL: \url{http://dl.acm.org/citation.cfm?id=982792.982822}.

\bibitem{KhanSocg21}
Arindam Khan, Arnab Maiti, Amatya Sharma, and Andreas Wiese.
\newblock On guillotine separable packings for the two-dimensional geometric
  knapsack problem.
\newblock In {\em 37th International Symposium on Computational Geometry
  ({SoCG})}, volume 189, pages 48:1--48:17. Schloss Dagstuhl - Leibniz-Zentrum
  f{\"{u}}r Informatik, 2021.
\newblock \href {http://dx.doi.org/10.4230/LIPIcs.SoCG.2021.48}
  {\path{doi:10.4230/LIPIcs.SoCG.2021.48}}.

\bibitem{KhanR20}
Arindam Khan and Madhusudhan~Reddy Pittu.
\newblock On guillotine separability of squares and rectangles.
\newblock In {\em Approximation, Randomization, and Combinatorial Optimization.
  Algorithms and Techniques ({APPROX/RANDOM})}, volume 176, pages 47:1--47:22.
  Schloss Dagstuhl - Leibniz-Zentrum f{\"{u}}r Informatik, 2020.
\newblock \href {http://dx.doi.org/10.4230/LIPIcs.APPROX/RANDOM.2020.47}
  {\path{doi:10.4230/LIPIcs.APPROX/RANDOM.2020.47}}.

\bibitem{khanna1998approximating}
Sanjeev Khanna, S.~Muthukrishnan, and Mike Paterson.
\newblock On approximating rectangle tiling and packing.
\newblock In {\em Proceedings of the Ninth Annual {ACM-SIAM} Symposium on
  Discrete Algorithms ({SODA})}, pages 384--393. {ACM/SIAM}, 1998.
\newblock URL: \url{http://dl.acm.org/citation.cfm?id=314613.314768}.

\bibitem{lewin2002routing}
Liane Lewin{-}Eytan, Joseph Naor, and Ariel Orda.
\newblock Routing and admission control in networks with advance reservations.
\newblock In {\em 5th International Workshop on Approximation Algorithms for
  Combinatorial Optimization ({APPROX})}, volume 2462, pages 215--228.
  Springer, 2002.
\newblock \href {http://dx.doi.org/10.1007/3-540-45753-4\_19}
  {\path{doi:10.1007/3-540-45753-4\_19}}.

\bibitem{marx2005efficient}
D{\'{a}}niel Marx.
\newblock Efficient approximation schemes for geometric problems?
\newblock In {\em 13th Annual European Symposium on Algorithms ({ESA})}, volume
  3669, pages 448--459. Springer, 2005.
\newblock \href {http://dx.doi.org/10.1007/11561071\_41}
  {\path{doi:10.1007/11561071\_41}}.

\bibitem{mitchell2021approximating}
Joseph S.~B. Mitchell.
\newblock Approximating maximum independent set for rectangles in the plane.
\newblock {\em CoRR}, abs/2101.00326, 2021.
\newblock Version 1.
\newblock URL: \url{https://arxiv.org/abs/2101.00326v1}.

\bibitem{DBLP:journals/corr/abs-2101-00326-v3}
Joseph S.~B. Mitchell.
\newblock Approximating maximum independent set for rectangles in the plane.
\newblock {\em CoRR}, abs/2101.00326, 2021.
\newblock Version~3.
\newblock URL: \url{https://arxiv.org/abs/2101.00326v3}.

\bibitem{MomkeW20}
Tobias M{\"{o}}mke and Andreas Wiese.
\newblock Breaking the barrier of 2 for the storage allocation problem.
\newblock In {\em 47th International Colloquium on Automata, Languages, and
  Programming ({ICALP})}, volume 168, pages 86:1--86:19. Schloss Dagstuhl -
  Leibniz-Zentrum f{\"{u}}r Informatik, 2020.
\newblock \href {http://dx.doi.org/10.4230/LIPIcs.ICALP.2020.86}
  {\path{doi:10.4230/LIPIcs.ICALP.2020.86}}.

\bibitem{mustafa2014settling}
Nabil~H. Mustafa, Rajiv Raman, and Saurabh Ray.
\newblock Settling the apx-hardness status for geometric set cover.
\newblock In {\em 55th {IEEE} Annual Symposium on Foundations of Computer
  Science ({FOCS})}, pages 541--550. {IEEE} Computer Society, 2014.
\newblock \href {http://dx.doi.org/10.1109/FOCS.2014.64}
  {\path{doi:10.1109/FOCS.2014.64}}.

\bibitem{mustafa2010improved}
Nabil~H. Mustafa and Saurabh Ray.
\newblock Improved results on geometric hitting set problems.
\newblock {\em Discret. Comput. Geom.}, 44(4):883--895, 2010.
\newblock \href {http://dx.doi.org/10.1007/s00454-010-9285-9}
  {\path{doi:10.1007/s00454-010-9285-9}}.

\bibitem{nielsen2000fast}
Frank Nielsen.
\newblock Fast stabbing of boxes in high dimensions.
\newblock {\em Theor. Comput. Sci.}, 246(1-2):53--72, 2000.
\newblock \href {http://dx.doi.org/10.1016/S0304-3975(98)00336-3}
  {\path{doi:10.1016/S0304-3975(98)00336-3}}.

\bibitem{pach2000cutting}
J{\'{a}}nos Pach and G{\'{a}}bor Tardos.
\newblock Cutting glass.
\newblock {\em Discret. Comput. Geom.}, 24(2-3):481--496, 2000.
\newblock \href {http://dx.doi.org/10.1007/s004540010050}
  {\path{doi:10.1007/s004540010050}}.

\end{thebibliography}

\section{A recursive partitioning for a 3-approximate set}\label{sec:3apx}

In this section, we construct a $228$-recursive partition for a set $\R'\subseteq \OPT'$, such that $|\R'|\ge \frac{1}{3}|\OPT'|$, where $\OPT'$ denotes the set of maximal rectangles in $S$. We first show how to construct this recursive partition to obtain $\R'$, using a more general Partitioning Lemma (Lemma \ref{lemma:partitionningLemma}) than for the $6$-approximation. Then, we prove the Partitioning Lemma, and finally we describe a charging scheme for the rectangles that are intersected during the partitioning process, and analyze it to show that $\R'$ is a $3$-approximate solution.

\subsection{Recursive Partitioning}

As before, we assume without loss of generality that at most half of the rectangles in $\OPT'$ are horizontally nested. Otherwise, it must be that at most half of the rectangles in $\OPT'$ are vertically nested, and we can for example rotate $\OPT'$ by $90$ degrees to obtain an equivalent instance $\OPT'$ where at most half of the rectangles are horizontally nested. 

In order to describe this recursive partition, we initialize the corresponding
tree $T$ with root $r$. We define now a recursive procedure that takes
as input a so far unprocessed vertex $v$ of the tree, corresponding
to some polygon $P=P_{v}$. It either partitions $P$ further (hence
adding children to $v$) or assigns a rectangle $R\in\OPT'$ to $v$
and does not add children to $v$. We denote by $\OPT'(P)\subseteq\OPT'$
the subset of rectangles of $\OPT'$ that are contained in $P$. 
If $\OPT'(P)=\emptyset$, then we simply delete $v$ from $T$. If $|\OPT'(P)|=1$
then we add the single rectangle in $\OPT'(P)$ to $\R'$, assign
it to $v$, and do not process $v$ further. 

Assume now that $|\OPT'(P)|>1$. 

\begin{figure}[h]
        \centering \includegraphics[page=19,width=0.9\textwidth]{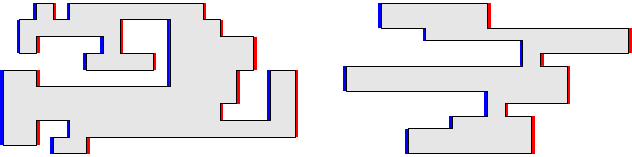}
    \caption{The left vertical edges (blue) and the right vertical edges (red) of an axis-parallel polygon. The polygon on the right is horizontally convex.}
    \label{fig:vertical_edges}     
\end{figure}

\begin{figure}[h]
    \centering \includegraphics[page=4,trim=80 110 60 60, clip, width=0.65\textwidth]{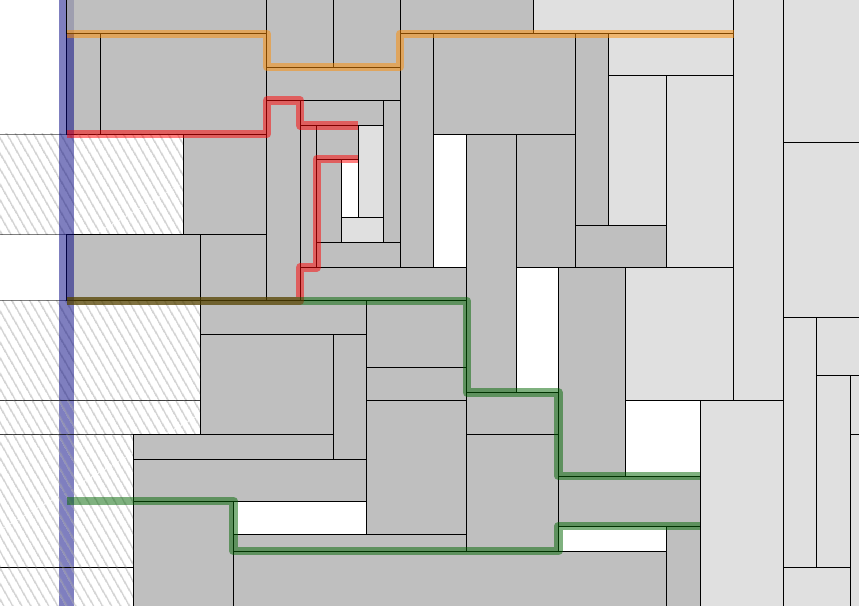}
    \caption{Some $5$-fences emerging from a left vertical edge (blue) on the boundary of an axis-parallel polygon.
        Shaded rectangles are the rectangles in $\protect\OPT'$ that intersect
        the boundary of $P$.
        White spaces indicate points that are not covered by rectangles of
        $\protect\OPT'$ contained in $P$. Darker gray rectangles correspond to rectangles in $\protect\OPT'(P)$
    that are $5$-protected. }
    \label{fig:5fences} 
\end{figure}
\arir{4th rectangle from left in the top row in Fig. 7, should be darker gray as it is 5-protected.}

We say that a curve $\gamma$ of the plane is \emph{$x$-monotone} (resp. \emph{$y$-monotone}) if for any vertical (resp. horizontal) line $\ell$, the intersection $\ell\cap \gamma$ is connected. 

Given an integral parameter $\tau\ge 1$ and a vertical edge $e$ of $P$, a $\tau$-fence \emph{anchored} on $e$ is a $x$-monotone curve defined by a sequence of at most $\tau$ horizontal and vertical line segments that is contained in $P$, has one endpoint on $e$ and intersects no rectangles of $\OPT'$. 
 See Figure \ref{fig:5fences}. 
 Denote by $F_\tau(P)$ the set of all fences
created in this way.

When we partition $P$, we will cut $P$ along a sequence of line segments such
that (i) no horizontal line segment intersects a rectangle in $\OPT'(P)$
and (ii) the interior of any line segment is disjoint from the interior of any $\tau$-fence. Therefore, the $\tau$-fences intuitively protect some rectangles in
$\OPT'(P)$ from being intersected by line segments defined in future iterations of the partition. This motivates the following definition.
\begin{defn}
    Given a polygon $P$, we say that a rectangle
    $R\in\OPT'(P)$ is \emph{$\tau$-protected in $P$} if there exists a vertical edge $e$ in $P$ such that the top side of $R$ is contained in a $\tau$-fence emerging from $e$ and the bottom of $R$ is contained in another $\tau$-fence also emerging from $e$. 

\end{defn}

\begin{obs}
Given an axis-parallel polygon $P$ and a parameter $\tau\ge 3$, if a rectangle $R\in \OPT'(P)$ has its top or bottom edge contained in a $(\tau-2)$-fence, then it is $\tau$-protected. 
\end{obs}
\begin{proof}
Without loss of generality, we assume that the top edge of $R$ is contained in a $(\tau-2)$-fence $f$ anchored on a point $p$ on a left vertical edge of $P$. Using $f$, the left edge of $R$ and the bottom edge of $R$, we can construct another fence an $x$-monotone sequence of at most $\tau$ line segments anchored on $p$ that contains the bottom edge of $R$. Thus, $R$ is $\tau$-protected. 
\label{obs:necessary_condition_tau_protected}
\end{proof}

We will ensure that no protected rectangle will be intersected
when we cut $P$, and the fences in $F_\tau(P)$ will help us to ensure
this. We apply the following lemma to $P$ with $\tau=7$.

\begin{restatable}[Partitioning Lemma]{lemma}{lemmapartition}
    \label{lemma:partitionningLemma}
    
    Let $\tau\ge 1$ be an integer. Let $P$ be a simple axis-parallel polygon with at most  $30\tau +18$ edges of integral coordinates, 
    such that $P$ contains at least two rectangles from $\OPT'$. Then there
    exists a sequence $C$ of line segments with integral coordinates such that: 
    \begin{enumerate}
        \item[(1)] $C$ is composed of at most $2\tau +1$ horizontal or vertical line segments
            that are all contained in $P$. 
        \item[(2)] $P\setminus C$ has exactly two connected components, and each of
            them is a simple axis-parallel polygon with at most $30\tau +18$ edges. 
        \item[(3)] There is a vertical line segment $\ell$ of $C$
            that intersects all the rectangles in $\OPT'(P)$ that are intersected
            by $C$. 
        \item[(4)] $\ell$ does not intersect any rectangle that is $\tau$-protected in $P$. 
    \end{enumerate}
\end{restatable}

We introduce an (unprocessed) child vertex of $v$ corresponding to
each connected component of $P\setminus C$ which completes the processing
of $P$. 

We apply the above procedure recursively to each unprocessed vertex
$v$ of the tree until there are no more unprocessed vertices. Let
$T$ denote the tree obtained at the end, and let $\R'$ denote the
set of all rectangles that we assigned to some leaf during the recursion. 

\begin{lemma}
    The tree $T$ and the set $\R'$ satisfy the following properties.
    \begin{enumerate}[label=(\roman*)]
        \item For each node $v\in T$, the horizontal edges of
            $P_{v}$ do not intersect any rectangle in~$\OPT'$. 
        \item $T$ is a $228$-recursive partition for $\R'$.
        \item If a rectangle $R\in\OPT'$ is $7$-protected in $P_{v}$ for some node
            $v\in T$, then 
            \begin{itemize}
                \item it is $7$-protected in $P_{v'}$ for each descendant $v'$ of $v$, 
                \item $R\in\R'$.
            \end{itemize}
    \end{enumerate}
    \label{lemma:partition_3apx}
\end{lemma}

\begin{proof}
    The first and second properties follow from the definition of
    the fences and Lemma~\ref{lemma:partitionningLemma}. The third
    property follows from the fact that for any $\tau\ge 1$,  if $f$ is a $\tau$-fence in $F(P_{v})$
    then, $f\cap P_{v'}$ is a $\tau$-fence of $F(P_{v'})$. 
\end{proof}

\subsection{Proof of the Partitioning Lemma}

In this section, we prove the partitioning Lemma used to construct the recursive partitioning of $S$. First, given an axis-parallel polygon $P$ with $k$ edges (and with no holes), we label the edges of $P$ according to their order in the boundary starting from an arbitrary edge in a clockwise fashion: $E(P)=e_1, \dots, e_k$. The distance between two edges $e_i$ and $e_j$ with $i\le j$ is $d(e_i,e_j)=\min\{j-i, k-j+i\}$. 

We will start by proving the following helpful lemma, \mm{which establishes a similar result as the one from Chazelle (\cite{Chazelle82}, Theorem 2) but now in the context of axis-parallel polygons. }

\begin{lemma}
Given an axis-parallel polygon $P$ with $k$ edges, there exists a vertical line segment $\ell$ such that $\ell$ is contained in $P$ and its endpoints are respectively contained in two horizontal edges of distance at least $k/3$.
\label{lemma:cut_polygon}
\end{lemma}

\begin{proof}
Let $\mathcal{L}$ be the set of vertical line segments contained in $P$ and with both endpoints on $P$'s boundary. Notice that a vertical line segment $\ell\in \mathcal{L}$ may contain some of $P$'s vertical edges. 
We also denote by $e_T$ and $e_B$ the horizontal edges of $P$ in which the top and the bottom endpoint of $\ell$ are respectively contained. 
Finally we denote $d(\ell):=d(e_B, e_T)$. 

To prove the lemma, we must show that there exists $\ell\in\mathcal{L}$ such that $d(\ell)\ge k/3$. 
For this, consider any $\ell\in\mathcal{L}$ such that $d(\ell)<k/3$. We construct, via a complete case analysis,  another $\ell'\in\mathcal{L}$ such that $d(\ell')>d(\ell)$.  See Figure \ref{fig:n_over_3}. 

\begin{figure}
   \centering \includegraphics[page=22, width=0.99\textwidth]{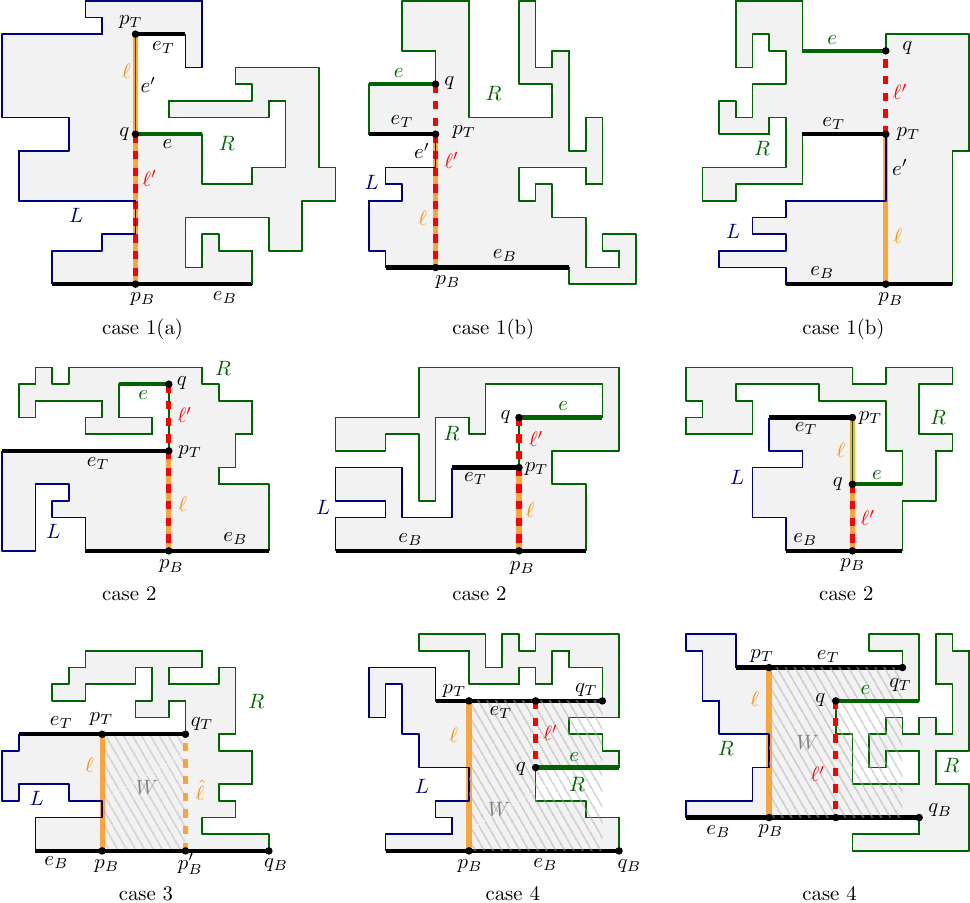}
\caption{Different cases in the proof of Lemma \ref{lemma:cut_polygon}. The line $\ell\in \mathcal{L}$ is displayed in orange, between points $p_B$ and $p_T$. On the polygon, edges from $L$ are shown in blue, while edges in $R$ are shown in green. The cut $\ell'$ such that $d(\ell')>d(\ell)$ is shown with a dashed thick red line.  }
\label{fig:n_over_3}
\end{figure}

\mm{We orient the edges of $P$ in a clockwise fashion. In this order, we can partition $P$'s boundary as: $E(P)=e_B\cup L \cup e_T\cup R$. }
W.l.o.g. we assume that $|R|\ge |L|$. In particular, we have $d(\ell)=|L|+1$. 

Let $p_T=\ell\cap e_T$ and $p_B=\ell\cap e_B$ denote the top and bottom endpoint of $\ell$ respectively. Similarly as the left and right vertical edges, we define the top horizontal edges (resp. bottom horizontal edges) of $P$ to any horizontal edge $e$, that contains in its interior a point $(x,y)$ such that the point $(x,y-1/2)\in P$  (resp. $(x,y+1/2)\in P$). 
\begin{enumerate}
\item[(case 1)] If $e_T$ is a bottom horizontal edge of $P$, then since $\ell$ is contained in $P$, $p_T$ must be an endpoint of $e_T$. 
\begin{enumerate}
\item[(a)] If $p_T$ is the left endpoint of $e_T$, then $p_T$ must be the top endpoint of an vertical edge $e'$ of $P$. 
    Furthermore, $e'$ is strictly contained in $\ell$, i.e., the bottom endpoint $q$ of $e'$ is contained in $\ell$ strictly above $p_B$.  The vertex $q$ is contained in a horizontal edge $e\in R$ and we have $d(e_B, e)=d(e_B,e_T)+2$. Thus, the vertical line segment $\ell'$ with endpoints $p_B$ and $q$ is in $\mathcal{L}$ and we have $d(\ell')=d(\ell)+2$.
\item[(b)] Otherwise, $p_T$ is the right endpoint of $e_T$. In this case, let $q$ be the lowest point above $p_T$, with the same $x$-coordinate as $p_T$ that is contained in the boundary of $P$. The point $q$ lies in a horizontal edge $e\in R$. We claim that either $d(e_B, e)$ or $d(e, e_T)$ is greater than $d(e_B, e_T)$. For the sake of a contradiction, assume that both are smaller than or equal to $d(\ell)$. Then, it must be that $d(e_B, e_T)+ d(e_T, e) + d(e, e_B) = k$, which implies that $d(e_B, e_T)\ge k/3$. Thus, we have either $d(e_T, e)>d(\ell)$ --- in which case we define $\ell'$ as the vertical line segment with endpoints $q$ and $p_T$ --- or $d(e_B, e)>d(\ell)$, and we then define $\ell'$ as the vertical line segment with endpoints $q$ and $p_B$. In both cases, we have $\ell'\in \mathcal{L}$ and $d(\ell')>d(\ell)$. 
\end{enumerate}
\item[(case 1')] If $e_B$ is a top horizontal edge of $P$, then we proceed symmetrically as case 1 to obtain $\ell'\in\mathcal{L}$ such that $d(\ell')>d(\ell)$.  
\end{enumerate}
Otherwise, $e_T$ is a top horizontal edge and $e_B$ is a bottom horizontal edge. Let $q_T$ and $q_B$ denote  the right endpoint of $e_T$ and the right endpoint of $e_B$, respectively. 
\begin{enumerate}
\item[(case 2)] If $p_T=q_T$, then $q_T$ must be the top endpoint (or the bottom) of a vertical edge $e'$ of $P$. The other endpoint $q$ of $e'$ is contained in a horizontal edge $e\in R$. Let $\ell'$ be the vertical line segment joining $p_B$ and $q$. It is clear that $\ell'$ is in $\mathcal{L}$ and since $e\in R$, we have $d(\ell')=d(\ell)+2$. 
\item[(case 2')]  If $p_B=q_B$, then we proceed symmetrically as case 2 to obtain $\ell'\in\mathcal{L}$ such that $d(\ell')=d(\ell)+2$.  
\end{enumerate}
Otherwise, both $q_T$ and $q_B$ are strictly on the right of $\ell$. W.l.o.g. we assume that $q_T$ has $x$-coordinate smaller than or equal to the one of $q_B$. Now, define the closed rectangular area $W$ that has $\ell$ as left side and $q_T$ as top-right corner. 
\begin{enumerate}
\item[(case 3)] If $W$ intersects no horizontal edge $e\in R$, then in particular the right side of $W$ is contained in $P$. Let $p'_B$ denote the bottom-right corner of $W$. By assumption, $p'_B$ lies in $e_B$. We replace $\ell$ by the vertical line segment $\hat{\ell}$ between $q_T$ and $p'_B$. We have $\hat{\ell}\in \mathcal{L}$, and $d(\hat{\ell})=d(\ell)$. Then, $\hat{\ell}$ satisfies the requirements of case 2. We apply the same process as in case 2 to obtain $\ell'\in\mathcal{L}$ such that $d(\ell')=d(\hat{\ell})+2=d(\ell)+2$. 

\item[(case 4)] Otherwise, $W$ intersects some horizontal edge in $R$. Since $\ell$ is contained in $P$, each such horizontal edge must have its left endpoint in $W$. Let $q\in W$ be (any of) the leftmost left endpoint of a horizontal edge $e\in R$. The vertical line segment between $e_B$ and $e_T$ that passes by $q$ is contained in $P$: by the choice of $q$,  it does not intersect the interior of any horizontal edge in $R$, and since $\ell\subset P$, no horizontal edge in $L$ intersects the interior of $W$. Then, since $e\in R$,  we can apply the same argumentation as in case 1(b) and prove that either $d(e_B, e)$ or $d(e_B, e)$ is greater than $d(e_B, e_T)=d(\ell)$. If $d(e_T, e)>d(e_B, e_T)$, we define $\ell'$ to be the vertical line segment between $e_B$ and $q$. Otherwise, we define $\ell'$ to be the vertical line segment between $e_T$ and $q$. In both cases, $\ell'\in \mathcal{L}$ and $d(\ell')>d(\ell)$. 
\end{enumerate}
Thus, for any $\ell\in \mathcal{L}$ such that $d(\ell)<k/3$, we have constructed another line $\ell'\in \mathcal{L}$ such that $d(\ell')> d(\ell)$. Consequently, the element $\ell_0\in \mathcal{L}$ that maximizes $d(\ell)$ over all $\ell\in \mathcal{L}$ must satisfy $d(\ell_0)\ge k/3$. 
\end{proof}



We can now prove our Partitioning Lemma restated below.

\lemmapartition*

\begin{proof}
Let  $P$ be an axis-parallel polygon with $k\le 30\tau +18=15(2\tau +1)+3$ edges.
\begin{enumerate}
\item[{\bf (case 0)}] If $k\le 15(2\tau +1)-1$, then let $R$ be one of the rectangles of $\OPT'(P)$ with the leftmost left side. The leftward horizontal ray $h_T$ (resp. $h_B$) from the top-right corner of $R$ (resp. from the bottom-right corner of $R$) reaches $P$'s boundary without intersecting any rectangle in $\OPT'(P)$. Let $C$ be the union of $h_T$, $h_B$, and the right side of $R$. The cut $C$ intersects no rectangles in $OPT'(P)$ and $P\setminus C$ has at least two connected components.  Since $C$ has three line segments, each of these component has at most $15(2\tau +1)-1+3+1=15(2\tau +1)+3$ edges. 

\end{enumerate}

\begin{figure}[h]
   \centering \includegraphics[page=25, width=0.98\textwidth]{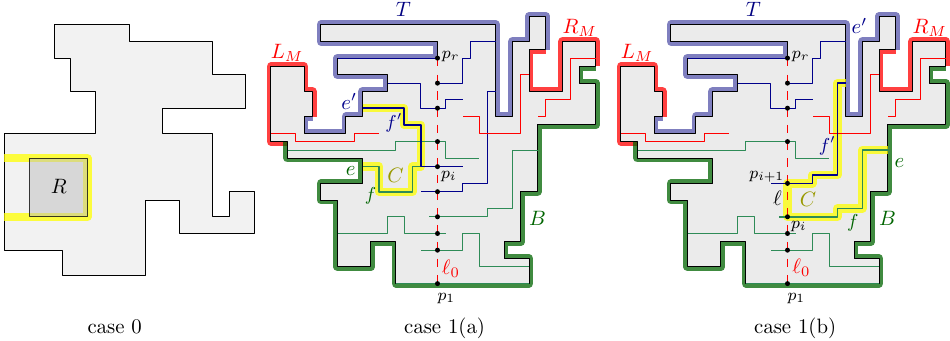}
\label{fig:cases01}
\end{figure}

 We assume now that $P$ has $k\ge 15(2\tau +1)$ edges. \mm{To construct a cut $C$ we will identify two $\tau$-fences $f_1$ and $f_2$ in $P$, and one vertical line segment $\ell$ that connects two points on $f_1$ and $f_2$ (with integral coordinates), and does not intersect the interior of any other $\tau$-fences. In particular, no $\tau$-protected rectangle is intersected by $\ell$. We will define the cut $C$ as a subset of $f_1\cup\ell\cup f_2$, so in particular it will consists of at most $2\tau+1$ line segments. We will choose $f_1$ and $f_2$ so that they are respectively anchored on vertical edges of distance at least $2\tau+2$ from each other. This implies that each of the two connected components of $P\setminus C$ will have at most $k\le  30\tau +18$ edges.}
 
  \mm{In some degenerate cases, the connected components of $P\setminus C$ may not be simple axis-parallel polygons. More precisely, each connected component may have some vertical edges that intersect and some horizontal edges that intersect, but since $C$ is contained in $P$ no vertical edge can cross an horizontal edge. This is due to the fact that $\tau$-fences may contain some edges of the polygon, and thus, the cut $C$ may intersect $P$'s boundary not only on its endpoints. See Figure \ref{fig:non-simple_components}. If this is the case, we will argue at the end of the proof that  there exists a subset $C'\subset C$ that intersects $P$'s boundary only on its endpoints, and so that $P\setminus C'$ has exactly two connected components that are simple axis-parallel polygon with at most $k\le 30\tau +18$ edges. }
\arir{This para was a bit confusing on the first reading, may be add a figure for this degenerate case later.}

Let $\ell_0$ be the vertical line segment given by Lemma \ref{lemma:cut_polygon}. Recall that the edges of $P$ are ordered in a clockwise fashion. In this order, we have: $E(P)=e_B\cup L \cup e_T\cup R$. By the choice of $\ell_0$ we know that both $L$ and $R$ have size at least $k/3-1\ge 5(2\tau +1)-1$. Since $|L|$ and $|R|$ are odd numbers, we have  $|L|\ge 5(2\tau +1)$ and  $|R|\ge 5(2\tau +1)$. We further partition each of these groups into three subgroups according to the clockwise ordering: $L=L_B\cup L_M\cup L_T$ and $R=R_T\cup R_M\cup R_B$ (See Figure \ref{fig:partition_boundary}), such that 
\begin{itemize}
\item The middle-left group $L_M$  and the middle-right group $R_M$ contain each at least $2\tau+1$ edges. 
\item The bottom-left group $L_B$, the top-left group $L_T$, the bottom-right group $R_B$ and the top-right group $R_T$ contain at least $2(2\tau+1)$ edges each. 
\end{itemize}

\begin{figure}
   \centering \includegraphics[page=24, width=0.6\textwidth]{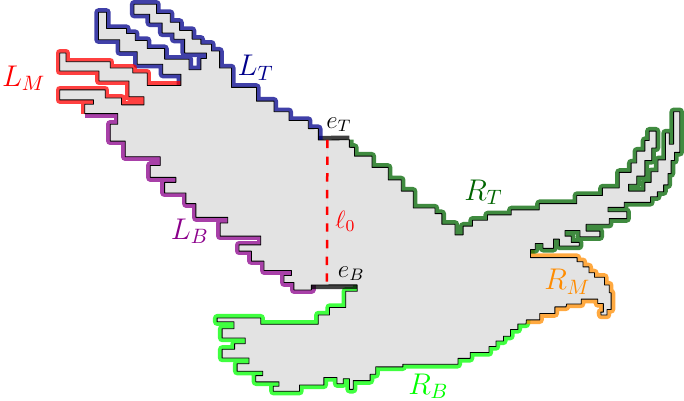}
\caption{The partition of $P$'s boundary into sets $\{e_B\},L_B,L_M,L_T,\{e_T\}, R_T, R_M,R_B$.  Middle groups $L_M$ and $R_M$ each contains at least $2\tau+1$ edges, while other groups ($L_T, L_B, R_B$ and $R_T$) each contains at least $2(2\tau +1)$ edges.}
\label{fig:partition_boundary}
\end{figure}

 Let $p_1, p_2, \dots, p_r$ the set of points of $\ell_0$, ordered from bottom to top, such that \ari{each $p_i$, $i \in [r]$} is contained in a $\tau$-fence. Let $E_i$ be the set of vertical edges $e$ such that there exists a \ari{$\tau$-fence} anchored in $e$ that contains $p_i$. Finally, we define $J=\cup_{i=1}^r E_i$. 

\emph{Remark.} Notice that the edge $e_B$ is covered by two $1$-fences $f$ and $f'$ (that are also $\tau$-fences), such that $f$ (resp. $f'$) is anchored on the vertical $e\in L_B$ (resp. $e'\in R_B$) that is incident to $e_B$. In particular, we have that $p_1\in e_B$, $E_1\cap L_B\neq \emptyset$ and $E_1\cap R_B\neq \emptyset$. Similarly, we remark that  $p_r\in e_T$, $E_r\cap L_T\neq \emptyset$ and $E_r\cap R_T\neq \emptyset$.

 \begin{enumerate}
 \item[{\bf(case 1)}] First assume that $(L_M\cup R_M)\cap J=\emptyset$, \ari{i.e., none of the fences anchored on edges in $L_M \cup R_M$ intersects $J$}. We define $B=R_B\cup e_B\cup L_B$ and $T=L_T\cup e_T\cup R_T$. For any edges $e\in B$ and $e'\in T$, we have $d(e,e')\ge \min(|L_M|, |R_M|)+1\ge 2\tau +2$. 
 \begin{enumerate}
 \item[{\bf (a)}] If there is an index $i$ such that $E_i\cap B\neq \emptyset$\madr{Is this $E_{i}$?} and $E_i\cap T\neq \emptyset$ then it means that $p_i$ is contained in a fence $f$ anchored on a edge $e$ of $B$ and $p_i$ is also contained in one fence $f'$ anchored on an edge $e'$ of $T$. Thus, there exists a cut $C\subseteq f\cup f'\subset P$ that is a sequence of at most $|f|+|f'|\le 2\tau$ line segments, that connects two points of the boundary contained on edges at distance at least $2\tau +2$ from each other\footnote{Notice that it is possible that $C\subsetneq  f\cup f'$, for instance in the case when $f$ and $f'$ share a point distinct than $p_i$.}\madr{What do you mean $P\subsetneq f \cup f'$?\ari{I think it should be $C\subsetneq f \cup f'$. Mathieu: Yes, thanks!}}. Thus, each connected components of $P\setminus C$ has at most $k-1\le 15(2\tau +1)+2$ edges. Since $C$ is a subset of two fences, it does not intersect any rectangle in $\OPT'(P)$. 

 \item[{\bf (b)}] If there is no such index, then each \ari{$E_i$} is either contained in $B$ or in $T$. In particular we obtain that $E_1\subseteq B$ and $E_r\subseteq T$. 
  Thus, there exists $1\le i\le r-1$ such that $E_i\subseteq B$ and  $E_{i+1}\subseteq T$. Thus, let $\ell$ be the vertical line segment between $p_i$ and $p_{i+1}$. The point $p_i$ is contained in a fence $f$ anchored on a edge $e\in B$ and $p_{i+1}$ is contained in a fence $f'$ anchored on an edge $e'\in T$. In particular, we have $d(e,e')\ge 2\tau+1$. Thus, there exists a sequence $C\subseteq f\cup \ell\cup  f'$ of at most $2\tau +1$ line segments that connects $e$ and $e'$, and  since $d(e,e')\ge 2\tau+1$, each component of $P\setminus C$ has at most $k\le 15(2\tau+1)+3$ edges. 
  
  We now check that conditions (3) and (4) hold. Since $f$ and $f'$ are $\tau$-fences, it is clear that only $\ell$ may intersect some rectangles in $\OPT'(P)$. Then, for the sake of a contradiction, assume that $\ell$ intersects a rectangle $R\in\OPT'(P)$ that is $\tau$-protected. There exists two $\tau$-fences $g$ and $g'$, both anchored on the same edge $e''$, such that the top edge of $R$ is contained in $g$ and the bottom side of $R$ is contained in $g'$.  Since $\ell$ intersects $R$, it must intersect the interior of $g$ and $g'$. This implies that $p_i\in g'\cap \ell$ and $p_{i+1}\in g\cap \ell$, and in particular we have $e''\in E_i\cap E_{i+1}$ and also $e''\notin (L_M\cup R_M)$. Thus, $e''\in T\cup B$, which brings a contradiction with the fact that $E_i\subseteq B$ and $E_{i+1}\subseteq T$. Therefore, $C$ does not intersect any $\tau$-protected rectangle. 
 \end{enumerate}

 \item[{\bf (case 2)}]
Now consider now the case $L_M\cap J\neq \emptyset$ or $R_M\cap J\neq \emptyset$; by symmetry, we can assume that $L_M\cap J\neq \emptyset$.  
 Let $f$ be the fence anchored on an edge $e\in L_M$, with the rightmost endpoint $p$ among all $\tau$-fences anchored on $L_M$. In particular, $f$ must intersect $\ell_0$.  Let $q$ denote the lowest point on the upward vertical ray starting from $p$, such that $q$ is contained in the interior of a $\tau$-fence $g$. We denote $e_g$ the edge on which $g$ is anchored. We also denote $\ell$ the vertical line segment between $p$ and $q$ (we may have $q=p$). 
Similarly, let $\hat{q}$ denote the highest point on the downward vertical ray $\hat{\ell}$ from $p$, such that $\hat{q}$ is contained in the interior of a $\tau$-fence $\hat{g}$ anchored on an edge $e_{\hat{g}}$ (we may also have $\hat{q}=p$). 
By the choice of $f$, we know that both $e_g$ and $e_{\hat{g}}$ are not in $L_M$. 

\begin{figure}
   \centering \includegraphics[page=26, width=0.88\textwidth]{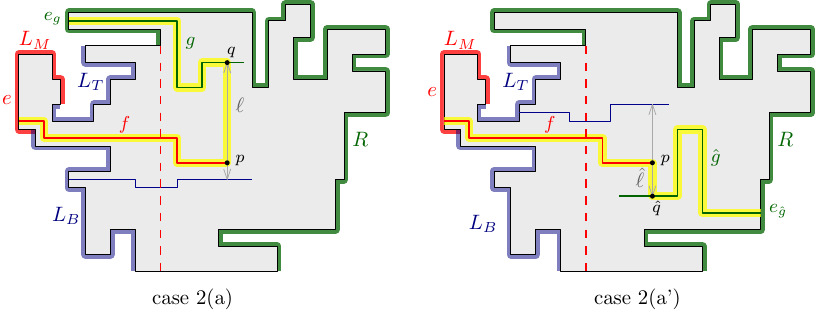}
\label{fig:case2a}
\end{figure}
 \begin{enumerate}
 \item[{\bf (a)}] If $e_g\in R$, then the distance between $e$ and $e_g$ is at least $\min(|L_T|,|L_B|)+1\ge 2\tau+2$. Thus, there exists a cut $C\subseteq f\cup \ell\cup  g$ with at most $2\tau + 1$ line segments, that satisfies the properties (1)-(3) of the Partitioning Lemma.  To prove that the last property also holds, we can use a similar argumentation as in case 1(b). Here, assuming that $\ell$ intersects a $\tau$-protected rectangle would imply a contradiction with the definition of $q$ as the \emph{lowest} point above $p$ that is contained in the interior of a $\tau$-fence. 
\item[{\bf (a')}] Similarly, if $e_{\hat{g}}\in R$, then we can define a cut $C \subseteq f\cup\hat{\ell} \cup\hat{g}$. 
\item[{\bf (b)}] Consider now the case where both $e_g$ and $e_{\hat{g}}$ are in $L$. Recall that, by definition of $f$, we know that $e_g$ and  $e_{\hat{g}}$ are not in $L_M$. 
 
 \begin{figure}
   \centering \includegraphics[page=27, width=0.88\textwidth]{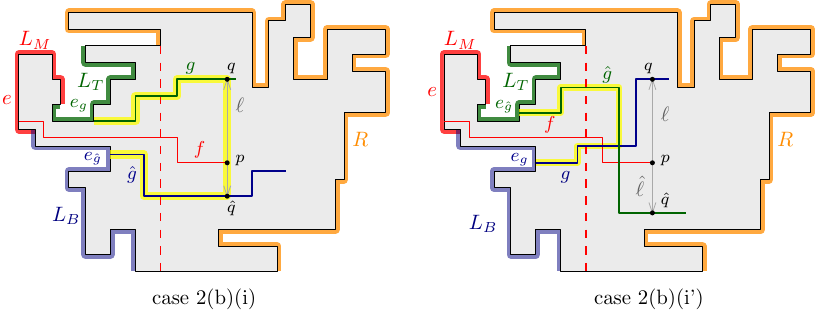}
\label{fig:case2bi}
\end{figure}

\begin{enumerate}
\item[{\bf (i)}] If $e_{\hat{g}}\in L_B$ and $e_g\in L_T$ then $d(e_{\hat{g}},e_g)\ge \min(|L_M|,|R|)+1\ge 2\tau+2$, and we define $C\subseteq g \cup \ell \cup \hat{\ell} \cup \hat{g}$. 
\item[{\bf (i')}] If $e_{\hat{g}}\in L_T$ and $e_g\in L_B$ then we also have $d(e_{\hat{g}},e_g)\ge \min(|L_M|,|R|)\ge 2\tau+2$.  In this case, $g$ and $\hat{g}$ must intersect and thus $C$ consists of the proper subset of $g\cup \hat{g}$ that connects $e_g$ and $e_{\hat{g}}$. 
\item[{\bf (ii)}] In the remaining case, $e_g$ and $e_{\hat{g}}$ are both in $L_T$ or both in $L_B$. These two cases are symmetrical so we only describe how to proceed in the former case. Let $\ell'$ denote the downward vertical ray from $f$'s endpoint $p$ to a point $q'$ on the boundary of $P$.  
 Let $q'_1,\dots, q'_{r'}$ be the points on $\ell'$, from top to bottom, such that each $q'_i$ is contained in \emph{ the interior} of a $\tau$-fence. Let $E'_1, \dots, E'_{r'}$ be sets of edges from $P$ such that $e'\in E'_i$ if there exists a fence anchored in $e'$ that contains $q'_i$. 
 We know that $E'_{r'}\cap (R\cup \{e_B\})\neq \emptyset$, and by definition of $f$ we also have that, for all $i\ge 1$,  $E'_i\cap L_M=\emptyset$. 
 
 We know that $q'$ lies on an horizontal edge contained in $R\cup\{e_B\}\subset E(P)\setminus(L_T\cup L_M)$.  
 Additionally, our assumption $e_{\hat{g}}\in L_T$ implies that there is an index $i_0$ such that $E'_{i_0}\cap L_T\neq \emptyset$. 
 Thus, there exists an index $i$, with $i_0\le i \le r'-1$ such that (a) $E'_i\cap L_T\neq \emptyset$ and (b) $(E'_i\cup E'_{i+1})\cap (E(P)\setminus (L_T\cup L_M))\neq\emptyset$. 
   We denote $g'$ the $\tau$-fence anchored on an edge $e'\in L_T$ such that $g'$ contains $q'_i$, and $g''$ the $\tau$-fence anchored on an edge $e''\in E(P)\setminus (L_T\cup L_M)$ such that $g''$ contains $q'_i$ or $q'_{i+1}$.
   
   We partition $L_T=L_T^-\cup L_T^+$ such that $L_T^+$ and $L_T^-$ are both continuous fraction of the boundary of $P$, the set $L_T^-$ is incident to $L_M$, and such that $|L_T^-|\ge \lfloor |L_T|/2\rfloor \ge 2\tau +1$ and $|L_T^+|\ge \lceil |L_T|/2\rceil\ge 2\tau +1$. We distinguish two subcases.
   
 
 \begin{figure}
   \centering \includegraphics[page=28, width=0.98\textwidth]{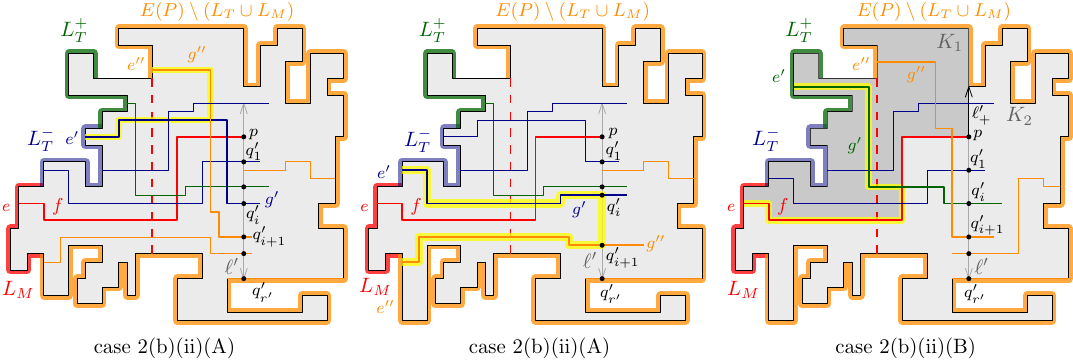}
\label{fig:case3bii}
\end{figure}

   \begin{enumerate}
   \item[{\bf (A)}] If $e'\in L_T^-$, then  $d(e',e'')\ge \min(|L_T^+|,|L_M|)+1\ge 2\tau +2$. Then, we define a cut $C\subseteq g' \cup g'' \cup \ell_i$ where $\ell_i$ is the vertical line segment between $q'_i$ and $q'_{i+1}$. In the case where $q_i\in g''$, the fences $g'$ and $g''$ intersects (on $q'_i$) (\ari{see the left figure for case 2(b)(ii)(A)}), so we can define $C$ as a subset of $g'\cup g''$.  
   
   \item[{\bf(B)}] Otherwise, we have $e'\in L_T^+$. Then $d(e',e)\ge \min(|L_T^-|,1+|R|+1+|L_B|)+1\ge 2\tau +2$. 
We claim $f$ and $g'$ must intersect. Let $\ell'_+\subset P$ denote the vertical line segment that connects $p$ (the right endpoint of $f$) to a point of $P$'s boundary, located above $p$. The set $f\cup \ell'_+$ separates $P$ into at least two connected components, among which there is one component $K_1$ that contains $e'$,  and another component $K_2$ that contains $\ell'$. 

Therefore, we know that $g'$ intersects the interior of $K_1$ and the interior of $K_2$ (since $g'$ intersects $\ell'$ on $q'_i$). Since $g'$ is connected and is contained in $P$, it must intersect $f\cup \ell'_+$. If it intersects $f$, the desired claim is proved. Otherwise there is a point $q_+\in \ell'_+\cap g'$. Since $g'$ was assumed to be $x$-monotone, the whole vertical line segment that connects $q'_i$ and $q_+$ is contained in $g'$. In particular $p\in g'$ which means that $g'\cap f\neq \emptyset$. 

Thus, we can define a cut $C\subset f \cup g'$ that have length at most $2\tau$ and connects edges $e$ and $e_g$, with $d(e,e_g)\ge 2\tau +2$. 
   \end{enumerate}
\end{enumerate}  
 \end{enumerate} 
  \end{enumerate}
  
  \mm{
At this stage, we have constructed a sequence $C$ of at most $2\tau+1$ line segments (with integral coordinates), that connects two edges $e$ and $e'$ on $P$'s boundary that are at distance at least $d(e,e')\ge2\tau+2$ from each other, and such that $P\setminus C$ has two connected components with at most $30\tau +18$ edges each. Therefore, $C$ satisfies property (1) of the Partitioning Lemma, and by construction it satisfies also properties (3) and (4).  If the two connected components of $P\setminus C$ are simple polygons, then such a cut satisfies the properties of the lemma. Otherwise, the interior of $C$ must intersect $P$'s boundary. 
See Figure \ref{fig:non-simple_components}. 
Let $C_1, \dots, C_r$ denote the maximal subpaths of $C$ whose interior are contained in the interior of $P$. In particular, the endpoints of each $C_i$ are on $P$'s boundary. So for each index $i$, $P\setminus C_i$ has exactly two connected components that are simple polygons. We claim that there exists an index $i$, such that these two connected components have at most $k\le 30\tau +18$ edges each. }
\arir{As I mentione earlier, perhaps a figure will be helpful here.}

\mm{For a contradiction assume that for each $i$, making the cut $C_i$ strictly increases the complexity. This means that the endpoints of $C_i$ are respectively on edges $e_i$ and $e'_i$ such that $d(e_i,e'_i)\le c_i$ where $c_i$ is the number of line segments of $C$ that are \emph{contained} in $C_i$. 
 Therefore, if $c\le 2\tau +1$ denotes the total number of line segments in $C$, we obtain the following contradiction:
$$
2\tau+2\le d(e,e') \le \sum_{i=1}^r d(e_i,e'_i)\le \sum_{i=1}^r c_i \le c \le 2\tau +1. 
$$
Therefore, there exists an index $i$ such that $c_i\le d(e_i,e'_i)-1$, which implies that $P\setminus C_i$ has exactly two connected components that are simple axis-parallel polygon and have at most $k\le 30\tau +18$ edges each. This finishes the proof. 
}

\begin{figure}
\centering
\includegraphics[page=32, width=0.8\textwidth]{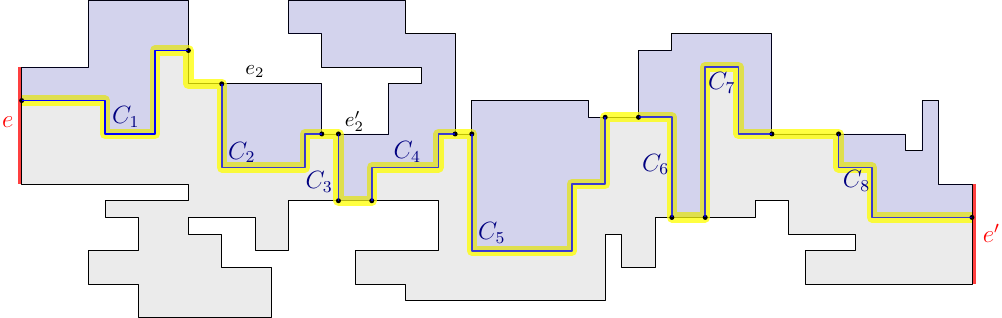}
\caption{Our cut $C$ (yellow) may partition a simple axis-parallel polygon $P$  into two connected components (blue and gray) that are not simple polygons. In this case, we identify one subpath $C'\subset C$ whose interior lies inside $P$ (and thus cuts $P$ into two simple polygons), such that cutting along $C'$ does not increase the complexity. For instance here,  $C'$ may be any of $\{C_3,C_4,C_6,C_7, C_8\}$. On the other hand, cutting $P$ along $C_1, C_2$ of $C_5$ would create one simple connected component with strictly more edges than $P$. Indeed, each such subpath $C_i$ ($i\in \{1,2,5\}$) connects two edges $e_i,e'_i$ of the boundary, that are at distance $d(e_i,e'_i)\le c_i$, where $c_i$ is the number of line segments of $C$ that $C_i$ contains. For instance, $2=d(e_2, e'_2)\le c_2=3$. We show that there exists always a subpath $C_i\subset C$ such that $d(e_i, e'_i)>c_i$. }
\label{fig:non-simple_components}
\end{figure}
\end{proof}

\subsection{Charging scheme and analysis.}
In this section we show that $|\R'|\geq |\OPT|/3$. We use a more sophisticated charging scheme than the one we previously described in Section~\ref{subsec:simple-recursion}. We still follow the same idea of charging non-horizontally nested rectangles that get intersected during the recursive partitioning to corners of rectangles that are not yet intersected. Here either we charge additional corners (that are not necessarily seen by the rectangle that is intersected), or we charge horizontally nested rectangles, and count them as ``saved''.  This will allow us to charge non-horizontally nested rectangle with a fractional charge of at most $\frac{1}{2}$. 

\paragraph{The charging scheme.}
Consider an internal node $v$ of the tree and let $\ell_{v}$ be the corresponding line segment $\ell$ defined above for partitioning $P_{v}$ as in Lemma~\ref{lemma:partitionningLemma}. Assume that there is a rectangle $R\subset P_{v}$ that is intersected by $\ell_{v}$ and that is not horizontally nested. By Lemma \ref{lemma:partitionningLemma} we know that $R$ is not $7$-protected in $P_{v}$ and, in particular, it is not protected (by a line fence) in $P_{v}$. Thus, Lemma \ref{lemma:charging_process} implies that $R$ sees at least one corner of some rectangle in $\OPT'(P_{v})$ on each side. For each such rectangle $R$, we introduce $4$ tokens that correspond to a fractional charge of $\frac{1}{4}$ each (we will later refer to these tokens as belonging to $R$). We distribute two tokens on the right and two on the left. Let us describe the token distribution on the right and the token distribution to the left will be symmetric. See Figure \ref{fig:token_distribution}.
\begin{figure}
   \centering \includegraphics[page=29, width=0.9\textwidth]{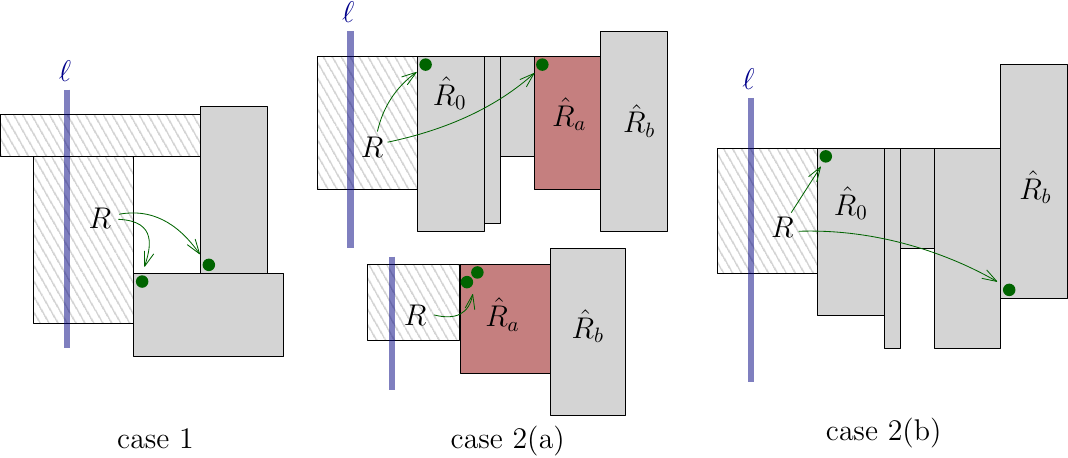}
\caption{Token distribution strategy for the charging scheme.}
\label{fig:token_distribution}
\end{figure}

Since $R$ is not horizontally nested, it sees at least one corner on its right. 
\begin{enumerate}
 \item If $R$ sees more than one corner, put one token on two of those corners each (choose any two such corners arbitrarily if needed). 
\item If $R$ sees only one corner $c$ of a rectangle, say $\hat R_{0}\in \OPT'$ on its right, then we give one token to $c$. To identify the second corner that will receive the second token from $R$, we first observe the following.

\begin{claim}
If $c$ is the top-left corner (resp. bottom-left corner) of $\hat R_{0}$, then the $y$-coordinates of the top edges (resp. the bottom edges) of $R$ and $\hat R_{0}$ are equal and the $x$-coordinates of the right edge of $R$ and left edge of $\hat R_{0}$ are equal. 
\end{claim}
\begin{proof}
Let $y_{b}$ and $y_{t}$ denote the $y$-coordinates of the bottom
and top edge of $R$, respectively, and let $y'_{b}$ and $y'_{t}$ denote the $y$-coordinates
of the bottom and top edge of $\hat R_{0}$, respectively. Additionally, let $x_{r}$ denote
the $x$-coordinate of the right edge of $R$ and let $x'_{l}$ denote
the $x$-coordinate of the left edge of $\hat R_{0}$. Observe that in the proof of Lemma~\ref{lemma:charging_process}, the only cases when $R$ sees only one corner to its right are case $1$ when $y_{t}=y'_{t}$ and $y'_{b}< y_b$, and case $2a$ when $y_{t}<y'_{t}$ and $y_{b}=y'_{b}$. In both these cases we have that $x_{r}=x'_{l}$, which completes the proof of the claim.   
\end{proof}

We assume without loss of generality that the unique corner $c$ that $R$ sees on its right is the top-left corner of $\hat R_{0}$, as the other case can be handled symmetrically. Let $h$ be the rightwards horizontal ray from the top-right vertex $p$ of $R$ and let $\hat R_{b}\in\OPT'(P_{v})$ be the first rectangle such that $h$ intersects the interior of the left edge of $\hat R_{b}$. Notice that such a rectangle must exist, as otherwise the top edge of $R$ would be contained in a 1-fence emerging from a right vertical edge of $P$. This would imply that $R$ is $3$-protected and in particular, it would be $7$-protected, contradicting property (4) of Lemma~\ref{lemma:partitionningLemma}. Let $p'$ be the intersection of $h$ with the left edge of $\hat{R}_b$. Let $H^r_R$ be the set of all rectangles in $\OPT'(P_{v})$ that have their top edge contained in the horizontal line segment $pp'$.
\begin{enumerate}
\item Suppose there exists in $H^r_R$ a rectangle $\hat R_{a}$ that is horizontally nested. Notice that then, the right edge of $\hat R_{a}$ has to be contained in the left edge of $\hat R_{b}$. In this case, give the second token to the top-left corner of $\hat R_{a}$.  In particular, in the case when $\hat R_{a}=\hat R_{0}$, the top-left corner of $\hat R_{a}$ has received two tokens from $R$. 
\item If none of the rectangles in $H^r_R$ are horizontally nested, then give the second token to the bottom-left corner of $\hat R_{b}$.
(In the case when the bottom edges of $R$ and $\hat R_{0}$ have the same $y$-coordinate, we charge the top-left corner of $\hat R_{b}$).
\end{enumerate}     
\end{enumerate}

We say that the corners of rectangles $\hat R_{0}, \hat R_{a}$ or $\hat R_{b}$ in $\OPT'$ that received a token from $R$ were \emph{charged by} $R$. 
We perform this token distribution for each non-horizontally nested rectangle $R\in OPT'$ that is intersected during the recursive partitioning. See Figure \ref{fig:example_token_distribution} for an example of this token distribution. 

\begin{figure}[h]
\centering
 \includegraphics[page=31,trim=80 180 45 70, clip, width=1\textwidth]{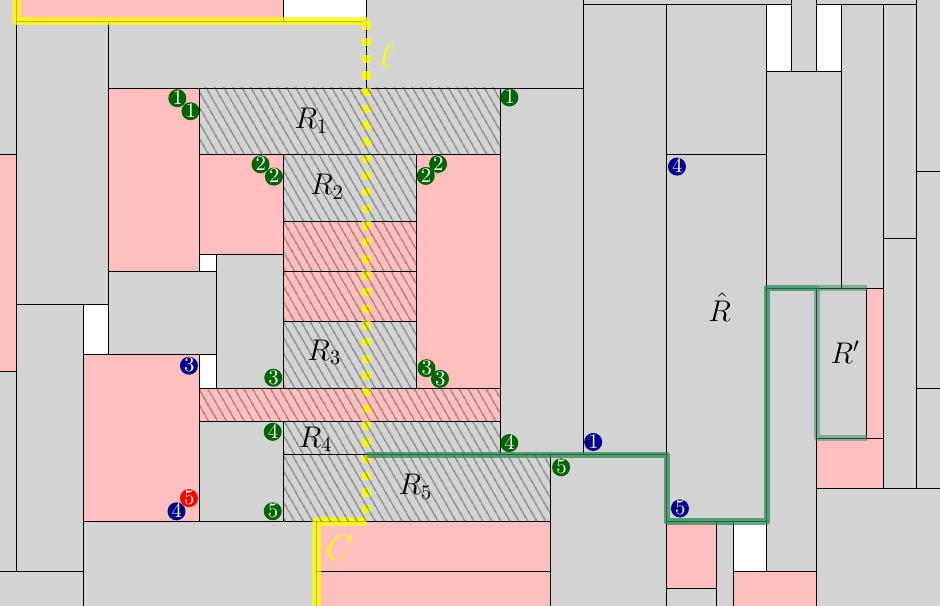}
\caption{An example of the charging process for rectangles intersected (tilling pattern) by the vertical line segment (dashed yellow line) of a cut (yellow) of a polygon $P$ during the partitioning process. Rectangles in $\OPT'(P)$ that are horizontally nested are displayed in red. For each rectangle $R_i$ ($1\le i \le 5$) that is intersected by $\ell$ and not nested horizontally, we introduce 4 tokens (marked by the number $i$); and we distribute two tokens to corners of rectangles of $\OPT'(P)$ on the right, and two on the left. In further iterations of the partitioning process, rectangles that has been charged will be $3$-protected (Lemma \ref{lemma:properties_chaging_scheme_bis}) and thus will never be intersected. Each corner of a rectangle is charged with at most one token if it is not horizontally nested, and with at most two tokens if it is horizontally nested (Lemma \ref{lemma:corner_charged_at_most_once_3}). Additionally, if the left corner of a rectangle is charged (for example the bottom-left corner of $\hat R$), then all the rectangles on the right of $\hat R$ that may potentially charged right corners of $\hat R$ (such as $R'$), will be protected  by $7$-fences (green line segments) emerging from $\ell$, and thus right corners of $\hat R$ will not be charged later (Lemma \ref{lemma:at_most_two_corners_charged_3}).
} 
\label{fig:example_token_distribution}
\end{figure}

\paragraph{Analysis of the charging scheme.} 
We now analyze our charging scheme. We show that following the charging scheme mentioned above, we can guarantee that the number of tokens on a non-horizontally nested rectangle and a horizontally nested rectangle is at most two and four respectively. We show this by first proving that there can be at most one token on a corner of a non-horizontal rectangle and at most two tokens on a corner of a horizontally nested rectangle. Then we prove that only one of the two corners of any horizontal edge of a rectangle can be charged.    

To prove these results, we first need to establish the following technical properties. 
  
\begin{lemma}
\label{lemma:properties_chaging_scheme_bis}
 Let $R\in \OPT'(P_{v})$ be a non-horizontally nested rectangle intersected during the partitioning of $P_v$ for some node $v\in T$, and let $c$ be the left corner of a rectangle $\hat{R}\in\OPT'(P_{v})$ charged by $R$. 
 Then we have that: 
 \begin{enumerate}
 \item[(1)] $c$ is seen by $R$ or by the rightmost rectangle on the left of $\hat{R}$ in $H_R^r$;
 \item[(2)] if $v'\in T$ denotes the child of $v$ such that $\hat{R}\subseteq P_{v'}\subset P_v$, then $\hat{R}$ and all the rectangles in $H^r_R$ are protected by $3$-fences emerging from a left vertical edge of $P_{v'}$. 
 \end{enumerate}
\end{lemma}
\begin{proof}
The proof of $(1)$ follows directly if $c$ is seen by $R$ or $\hat{R}\in H^{r}_{R}$ because the right most rectangle to the left of $\hat{R}$ in $H^{r}_{R}$ will belong to $\{R\}\cup H^{r}_{R}$ and it is easy to see that this rectangle will see the corner $c$ of $\hat{R}$ from the proof of Lemma \ref{lemma:charging_process}. If both conditions do not happen, then corner $c$ is charged via case $2(b)$ of the charging scheme. Let us assume without loss of generality that $c$ is the bottom left corner of $\hat{R}$. The rightward ray from the top right vertex $p$ of $R$ hits the interior of the left edge of $\hat{R}$ at $p'$ and the line segment $h$ joining $p$ and $p'$ contains the top edges of all the rectangles in $H^{r}_{R}$ by the way we find $\hat{R}$ and $c$ in the charging scheme. The fact that we charged $c$ implies that the rightmost rectangle $\hat{R}_{a}$ in $H^{r}_{R}$ is non-horizontally nested. Since the horizontal line segment $h'$ joining the top right vertex of $\hat{R}_{a}$ and $p'$ on the left edge of $\hat{R}$ does not intersect any rectangle in $\OPT'(P_{v})$ and also does not contain the top edges of any rectangle in $\OPT'(P_{v})$, we can say that $\hat{R}_{a}$ has to see the corner $c$ of $\hat{R}$ following a similar argument as in the proof of Lemma \ref{lemma:charging_process}.   

The proof of $(2)$ will follow from the above proof of $(1)$ because, if $c$ is seen by $R$ or $\hat{R}\in H^{r}_{R}$, then after partitioning $P_{v}$, the horizontal line segment connecting $c$ and the left boundary of $P_{v'}$ does not intersect any rectangle in $\OPT'(P_{v'})$. This means that $c$ is contained in a $1$-fence and hence $\hat{R}$ is $3$-protected. If $c$ is seen by the rightmost rectangle $\hat{R}_{a}$ in $H^{r}_{R}$, then we can say that $c$ is in a $3$-fence that basically connects $p,p'$ and $c$ in that order. Since $p'$ is in the interior of the left edge of $\hat{R}$, we can say that the top edge is also contained in a $3$-fence emerging from the same point as the $3$-fence that contained $c$, proving that $\hat{R}$ is $3$-protected. Rectangles in $H^{r}_{R}$ are protected by $3$-fences because their top edges (assuming $c$ is a top left corner) are contained in the $1$-fence emerging from the point of intersection of the top edge of $R$ with the part of the vertical cut $\ell_{v}$ partitioning $P_{v}$.     
\end{proof}

The last property of the claim ensures that a rectangle that is charged will be $3$-protected after being charged for the first time, and in particular using Lemma \ref{lemma:partition_3apx}(iii) we directly obtain the following property of our charging scheme. 

\begin{lemma}
If a rectangle $\hat R\in \OPT'$ is charged, then $\hat R\in \R'$. 
\label{lemma:charged_implies_protected_3} 
\end{lemma}

We now bound the number of tokens charged to a corner of a rectangle $\hat R\in \R'$. This bound depends on whether $\hat R$ is horizontally nested or not. 

\begin{lemma}
Let $\hat R\in \R'$ be a not horizontally nested (resp. horizontally nested) rectangle. Then, each of its corners is charged with at most {\bf one} token (resp. at most {\bf two} tokens).
\label{lemma:corner_charged_at_most_once_3}
\end{lemma}

\begin{proof}
We first distinguish two types of tokens. If the corner $c$ of a rectangle $\hat R\in \R'$, charged by a rectangle $R\in \OPT'$, is seen by $R$, then we say that the corresponding token is a \emph{direct} token, and that $c$ was \emph{directly charged by $R$}. Otherwise, we say that it is an \emph{indirect} token, and that $c$ was indirectly charged by $R$. For instance, in Figure \ref{fig:token_distribution}, the tokens received by $\hat R_{b}$ in case 2(b) and by $\hat R_{a}$ in case 2(a) top, are indirect tokens. In Figure \ref{fig:example_token_distribution}, direct tokens are shown in green, and indirect tokens are shown in red and blue. 


Since a corner of a rectangle in $\OPT'$ is seen by at most one rectangle in $\OPT'$, a corner cannot be \emph{directly} charged by two different rectangles. In particular, a rectangle $R$ charges a corner $c$ of a rectangle $\hat R$ with one token if either this corner is not the unique corner seen by $R$ on the same side as corner $c$ \ari{(see case 1 in Figure \ref{fig:token_distribution})} or $\hat R$ is non-horizontally nested, and with two tokens if $\hat R$ is horizontally nested and $c$ is the unique corner seen by $R$ to that side (this happens in the case 2(a), when $\hat R_a=\hat R_0$). Moreover, observe that a corner cannot be simultaneously charged directly and indirectly by the same rectangle. We now show that if a corner $c$ of a rectangle $\hat R\in \R'$ is directly charged by a rectangle $R\in \OPT'$, it cannot be also indirectly charged by another rectangle $R'\in \OPT'$.   

For a contradiction, assume that $c$ was directly charged by a rectangle $R$ intersected during the partitioning of $P_v$ for some node $v\in T$ and that $c$ was indirectly charged by a rectangle $R'\neq R$ that was intersected during the partitioning of $P_{v'}$ for some node $v'\in T$. We have either $P_v=P_{v'}$, $P_v\subset P_{v'}$ or $P_{v'}\subset  P_v$. 
Notice that from the first property of Lemma \ref{lemma:properties_chaging_scheme_bis}, it holds that $R\in H^r_{R'}$.  
\begin{itemize}
\item If $P_v=P_{v'}$, then $R$ and $R'$ were intersected by the same vertical line. This is a contradiction with the facts that $R$ is in $H^r_{R'}$ and that all the rectangles in $H^r_{R'}$ lie strictly on the right of $R'$. 
\item If $P_v\subset P_{v'}$, then since $R\in H^r_{R'}$, we know by Lemma \ref{lemma:properties_chaging_scheme_bis} that $R$ is $3$-protected in $P_v$ which contradicts the fact that it was intersected during the partitioning of $P_v$.
\item If $P_{v'}\subset  P_v$, then the leftwards horizontal ray from $c$ must intersect the boundary of $P_{v'}$ before intersecting any rectangles in $P_{v'}$, which contradicts the fact that $R'\in \OPT'(P_{v'})$.
\end{itemize}
Therefore, a corner of a rectangle in $\R'$ cannot be at the same time directly and indirectly charged.  

We now claim that a corner of a rectangle $\hat R\in \R'$ is charged with at most two indirect tokens, and that in this case, $\hat R$ must be horizontally nested. This will conclude the proof of the Lemma. 

To prove this fact, we first distinguish two types of indirect tokens: the token of type (a) charged to the horizontally nested rectangle $\hat R_{a}$ in case 2(a) of the charging process; and the token of type (b), charged to $\hat R_{b}$ in case 2(b). In Figure \ref{fig:example_token_distribution}, indirect tokens of type (a) and (b) are respectively shown in red and blue. 

We show that a corner $c$ of a rectangle $\hat R\in \R'$ cannot be charged with two indirect tokens of the same type. This implies that if a corner is charged with two indirect tokens, one of these indirect tokens is of type (a), which implies that $\hat R$ must be horizontally nested. 
\begin{enumerate}
\item[(a)] Suppose for a contradiction that a left corner $c$ of rectangle $\hat{R}$ has been charged with two tokens of type (a) by two distinct rectangles $R$ (intersected during the partitioning of $P_v$ for some $v\in T$) and $R'$ (intersected during the partitioning of $P_{v'}$ for some $v'\in T$). We assume without loss of generality that $c$ is the top-left corner, and that $P_{v'}\subseteq P_v$. Since $\hat{R}\in H_{R}^{r}$ and $\hat{R}\in H_{R'}^{r}$, there is minimal horizontal line segment connecting the top edges of $R,R'$ and $\hat{R}$ in $P_v$ that 
does not intersect any rectangle in $\OPT'(P_v)$. We obtain a contradiction similarly as before. If $P_{v'}=P_v$ then we obtain a contradiction with the fact that $R$ and $R'$ must be intersected by the same vertical line segment. Otherwise $R'$ was not intersected during the partitioning of $P_v$. In this case, if $R'$ is on the left of $R$, then the leftwards horizontal ray from $c$ must cross the boundary before reaching the top edge of $R'$; if not, then by Lemma \ref{lemma:properties_chaging_scheme_bis} $R'$ must be $3$-protected in $P_{v'}$. 

\item[(b)] 
Suppose a left corner $c$ of a rectangle $\hat{R}$ was charged with two indirect tokens of type (b) by two distinct rectangles $R$ and $R'$. We assume here without loss of generality that $c$ is a bottom-left corner. By Lemma \ref{lemma:properties_chaging_scheme_bis}, we know that $c$ is seen by a rectangle $\hat R_a\in H^r_R\cap H^r_{R'}$. Therefore, $R$ and $R'$ are on the left of $\hat R_a$ and there is an horizontal line segment that contains the top edges of $R, R'$ and $\hat R_a$. Thus, we can argue as in case (a) to obtain a contradiction. 
\end{enumerate}
This finishes the proof. 
\end{proof}

We now show that if a left corner of a rectangle $\hat R\in \R'$ is charged, then no right corner of this rectangle is charged, and vice-versa. Intuitively, by Lemma \ref{lemma:properties_chaging_scheme_bis}, we know that if a left corner of $\hat R$ has been charged, the top edge $\hat R$ will be contained in a 3-fence emerging from the left; by extending this 3-fence to 7-fences, we will protect all rectangles in $\OPT'$ that may potentially charge right corners of $\hat R$. 
\begin{lemma}
At most two corners per rectangle in $\R'$ are charged. 
    \label{lemma:at_most_two_corners_charged_3} 
\end{lemma}
\begin{proof}
Suppose that the left corner $c$ of a rectangle $\hat R$ was charged by a rectangle $R\in \OPT'$ intersected during the partitioning of $P_v$ for some $v\in T$. First, since we charge left corners on the right and right corners on the left, it is clear that the right corners of $\hat R$ cannot be charged by any rectangle intersected during the partitioning of the same polygon $P_v$. 

Assume for a contradiction that a right corner of $\hat{R}$ was charged by a rectangle $R'$ during the partitioning of $P_{v'}$ for some node $v'\neq v$ such that $\hat R \subset P_{v'}\subset P_v$. 
By Lemma \ref{lemma:properties_chaging_scheme_bis}, we know that
 the top edge of $\hat R$ is contained in a $3$-fence $f$ emerging from a left vertical edge of $P_{v'}$.  Moreover, since $\hat R$ was charged by $R'$, there exists a point $p'$ contained in the right edge of $\hat R$ such that there exists an horizontal line segment $h$ that contains $p'$ and the top or the bottom edge of $R'$, such that $h$ intersects no rectangles in $\OPT'$. We now construct a $5$-fence $f'$ anchored on a left vertical edge of $P_{v'}$ as the union of $f$, the fraction of the right edge of $\hat R$ between its top-right corner and $p'$, and the horizontal line segment $h$. See Figure \ref{fig:example_token_distribution}.  
 Therefore, the top or the bottom edge of $R'$ is contained in a $5$-fence, which implies that it is $7$-protected in $P_{v'}$, and thus contradicts the assumptions that it was intersected during the partitioning of $P_{v'}$. This shows that a rectangle can be charged only on one side. 
\end{proof}

Combining these three lemmas, we can lower bound the size of $\R'$.
\begin{lemma}
It holds that $|\R'|\ge |\OPT'|/3$. 
\end{lemma}

\begin{proof}
 Let us define $\I:= \OPT'\setminus \R'$, the set of rectangles that were intersected during the recursive partitioning process. Let $\OPT'_h\subset \OPT'$ denote the set of rectangles in $\OPT'$ that are horizontally nested. By our initial assumption, we have
\begin{equation}
|\OPT'_h|\le \frac{1}{2}|\OPT'|
\label{eq:1}
\end{equation} 
By Lemma \ref{lemma:charged_implies_protected_3}, we know that all the rectangles in $\OPT'$ that are charged are in $\R'$. Combining lemmas \ref{lemma:corner_charged_at_most_once_3} and \ref{lemma:at_most_two_corners_charged_3}, we know that each rectangle in $\R'\cap \OPT'_h$ is charged with at most $4$ tokens, i.e., for a total fractional charge of $1$, and each rectangle in $\R'\setminus \OPT'_h$ is charged with at most $2$ tokens, i.e., for a total fractional charge of $1/2$. Thus, we deduce that:
\begin{equation}
|\I\setminus \OPT'_h|\le |\R'\cap \OPT'_h|+\frac{1}{2}|\R'\setminus \OPT'_h|
\label{eq:2}
\end{equation}

Now we have, 
\begin{align*}
\frac{3}{2}|\I|&= \frac{3}{2}|\I\setminus \OPT'_h| + \frac{3}{2}|\I\cap \OPT'_h| \\
&\le |\R'\cap \OPT'_h|+\frac{1}{2}|\R'\setminus \OPT'_h| + \frac{1}{2}|\I\setminus \OPT'_h| + \frac{3}{2}|\I\cap \OPT'_h| \\
&=  |\R'\cap \OPT'_h| +\frac{1}{2} \big(|\OPT'|-|\OPT'_h|-|\I\setminus \OPT'_h|\big)+ \frac{1}{2}|\I\setminus \OPT'_h| + \frac{3}{2}|\I\cap \OPT'_h| \\
&= \frac{1}{2}|\OPT'|+ \frac{1}{2}|\R'\cap \OPT'_h|+ |\I\cap \OPT'_h|\\
&\le  \frac{1}{2}|\OPT'|+ |\OPT'_h|\\
&\le \frac{1}{2}|\OPT'|+ \frac{1}{2}|\OPT'|=|\OPT'|,
\end{align*}
where the first and the last inequalities follow  from equations (\ref{eq:2}) and (\ref{eq:1}), respectively. Thus, at most two thirds of the rectangles in $\OPT'$ are intersected during the recursive partitioning process, and in particular $|\R'|\ge \frac{1}{3}|\OPT'|=\frac{1}{3}|\OPT|$. 
\end{proof}

This concludes the proof of Theorem~\ref{theorem:3apx}.

\color{black}

\section{An improved $(2+\varepsilon)$ approximation}\label{sec:2_plus_eps_apx}
In this section, we build on the ideas from Section \ref{sec:3apx}.  We show that by using the same dynamic program and only making slight modifications to the definition of nesting and charging scheme, we can push the approximation factor arbitrarily close to $2$ by using a larger parameter $k$ of the dynamic program. Let us assume that $\frac{1}{\varepsilon}\in \mathbb{N}$ whenever we refer to $\varepsilon$ below.

\begin{defn}
	A rectangle $R\in \OPT'$ is said to be \emph{horizontally nice} (resp. \emph{vertically nice}) if $R$ sees the \emph{bottom left} (resp. \emph{top right}) corner of some rectangle in $\OPT'$ to its right (resp. bottom) or its bottom edge (resp. right edge) is contained in the boundary of $S$.   
\end{defn} 

The below Observation extends Lemma~\ref{lemma:charging_process} to a more general result without any assumptions about the rectangle whether it is protected or non-horizontally nested.
\begin{obs}
	\label{obs:lemma:charging_process_extension}
	For any rectangle $R\in \OPT'$, if the horizontal line segment
	$h$ that connects its top-right corner to a right vertical edge $e'$
	of $S$ 
	\begin{enumerate}
		\item Intersects a rectangle in $\OPT'$ or contains the top edge of a rectangle in $\OPT'$, then if $R'\in \OPT'$ be the rectangle intersected by $h$, or that has its top side
		intersected by $h$, that is the closest to $R$. Then either $R$ sees the bottom left corner of $R'$ or $y_{t}\leq y_{t}', y_{b}' < y_{b} ,x_{r}=x_{l}'$, where $y_{t},y_{t}',y_{b},y_{b}'$ are the $y$-coordinates of the top edges of $R, R'$ and bottom edges of $R, R'$, respectively, and $x_{r}, x_{l}'$ are the $x$-coordinates of the right and left edges of $R$ and $R'$, respectively.
		\item Does not intersect any rectangle nor contains the top edge of any rectangle in $\OPT'$, then the right edge of $R$ is contained in the right edge of $S$
	\end{enumerate}
\end{obs}
In the next proposition, we show that any rectangle in the optimal solution is either horizontally nice or vertically nice. 

\begin{prop}
	\label{prop:nice}
	A rectangle $R\in \OPT'$ has to be either horizontally nice or vertically nice. 
\end{prop}
\begin{proof}
	For contradiction, let us assume that there exists a rectangle $R\in \OPT'$ such that it is neither horizontally nice nor vertically nice. Consider a horizontal line segment $h$ that connects the top-right corner of $R$ to the right vertical edge of $S$. 
	If $h$ intersects some rectangle in $\OPT'$ or contains the top edge of some rectangle in $\OPT'$, then let $R'\in \OPT'$ denote the rectangle intersected by $h$, or that has its top side contained in $h$, that is closest to $R$. From Observation \ref{obs:lemma:charging_process_extension}, we know that either $R$ sees the bottom left corner of $R'$ or $y_{t}\leq y_{t}', y_{b}' < y_{b} ,x_{r}=x_{l}'$, where $y_{t},y_{t}',y_{b},y_{b}'$ are the $y$-coordinates of the top edges of $R, R'$ and bottom edges of $R, R'$, respectively, and $x_{r}, x_{l}'$ are the $x$-coordinates of the right and left edges of $R$ and $R'$, respectively.
	and if $h$ does not intersect any rectangle and does not contain the top edge of any other rectangle in $\OPT'$, then again by Observation \ref{obs:lemma:charging_process_extension}, we know that the right edge of $R$ is contained in the right edge of $S$ contradicting our assumption on $R$. So we have deduced that if a rectangle is not horizontally nice, then there exists a rectangle $R'$ (possibly $S$) such that the bottom right vertex of $R$ is strictly contained in the interior of a vertical edge of $R'$. Using symmetry, we can also prove that if a rectangle is not vertically nice, then there exists a rectangle $R''$(possibly $S$) such that the bottom right vertex of $R$ is strictly contained in the interior of a horizontal edge of $R''$. Since $R$ is neither horizontally nice nor vertically nice, we should have two rectangles that have a common interior point in $\OPT'$ -- contradicting our assumption about rectangles in $\OPT'$.    
\end{proof}
Using Proposition \ref{prop:nice}, we can assume without loss of generality that the number of horizontally nice rectangles is at least $|\OPT|/2$.
Now let us restate the partitioning lemma, where we will use $\tau = {4}/{\varepsilon}+3$. 
\lemmapartition*

We introduce an (unprocessed) child vertex of $v$ corresponding to
each connected component of $P\setminus C$ which completes the processing
of $P$. 

We apply the above procedure recursively to each unprocessed vertex
$v$ of the tree until there are no more unprocessed vertices. Let
$T$ denote the tree obtained at the end, and let $\R'$ denote the
set of all rectangles that we assigned to some leaf during the recursion. 
For any given $\varepsilon \geq 0$, we show that by using $\tau = {4}/{\varepsilon}+1$ in the partition lemma above, we can guarantee that $|\R'|\geq |\OPT|/(2+\varepsilon)$. 

\begin{lemma}
	The tree $T$ and the set $\R'$ satisfy the following properties.
	\begin{enumerate}[label=(\roman*)]
		\item For each node $v\in T$, the horizontal edges of
		$P_{v}$ do not intersect any rectangle in~$\OPT'$. 
		\item $T$ is a $120/\varepsilon+ 108$-recursive partition for $\R'$.
		\item If a rectangle $R\in\OPT'$ is $(4/\varepsilon +3)$-protected in $P_{v}$ for some node
		$v\in T$, then 
		\begin{itemize}
			\item it is $(4/\varepsilon+3)$-protected in $P_{v'}$ for each descendant $v'$ of $v$, 
			\item $R\in\R'$.
		\end{itemize}
	\end{enumerate}
	\label{lemma:partition_2pluseps_apx}
\end{lemma}
\begin{proof}
	The first and second properties follow from the definition of
	the fences and Lemma~\ref{lemma:partitionningLemma}(2) and $\tau = {4}/{\varepsilon}+3$. The third
	property follows from the fact that for any $\tau\ge 1$,  if $f$ is a $\tau$-fence in $F(P_{v})$
	then, $f\cap P_{v'}$ is a $\tau$-fence of $F(P_{v'})$. 
\end{proof}

\subsection{Charging scheme and analysis.}

\mm{
To describe our charging scheme, we first define a directed graph $H$ whose vertex set is $\OPT'$ and where there is a directed edge from a node corresponding to a rectangle $R\in \OPT'$ to that of $\hat{R}\in \OPT'$ if and only if $R$ sees the bottom-left corner of $\hat{R}$ on its right.  Since a corner of a rectangle can be seen by at most one rectangle on its left, the in-degree of each node in $H$ is at most one. Therefore, $H$ is a directed forest. Notice that in that forest, nodes with out-degree zero correspond to rectangles that are not horizontally nice. }

\mm{Before presenting formally the charging process, we need the following Claim that gives us a bridge between $\tau$-protection and distances in $H$. 
For the ease of the presentation, we identify a node in $H$ with its corresponding rectangle in $\OPT'$. 
\begin{claim}
Let $R$ and $R'$ be two rectangles in $\OPT'$ such that there exists an oriented path in $H$ of length $d$ that connects $R$ to $R'$. Then, there exists a sequence $f$ of at most $2d+1$ horizontal and vertical line segments such that $f$ intersects no rectangles in $\OPT'$ and such that the bottom edges of $R$ and $R'$ are respectively contained in the first and the last horizontal segment of $f$. 
\label{claim:distance_fence}
\end{claim}
\begin{proof}
We prove the claim by induction on $d$. First, if $d=0$, then $R=R'$ and we set $f$ as the bottom edge of $R$. Then, assume that $R$ and $R'$ are at distance $d\ge 1$ in $H$. Let $\hat{R}$ denote the rectangle incident to ${R'}$ on the directed path from $R$ to $R'$. By induction hypothesis, there exists a set $\hat{f}$ of at most $2(d-1)+1$ horizontal and vertical line segments that intersects no rectangle in $\OPT'$ and whose first and last horizontal segment respectively contain the bottom edge of $R$ and the bottom edge of $R'$. Since $\hat{R}$ sees the bottom-left corner of $R'$, there is a horizontal line segment $h$ that connects the bottom-right corner of $R'$ to a point $x$ on the right edge of $\hat{R}$, without intersecting any rectangle in $\OPT'$. We define $f$ as the union of $\hat{f}$, the vertical line segment from the bottom-right corner of $\hat{R}$ to $x$, and $h$. One can easily check that $f$ satisfies the expected properties. This terminates the proof of the claim. 
\end{proof}
}

\mm{
We now describe the charging process. Consider the tree $T$ obtained at the end of the recursive partitioning process. 
For each node $v\in T$ and each rectangle $R$ intersected by $\ell_v$ during the partitioning of $P_v$, we proceed as follows. Let $\gamma$ be any \emph{maximal} oriented path in $H$, rooted at $R$, such that all rectangles in $\gamma$ are contained in $P_v$. 
\begin{enumerate}
\item If $\gamma$ has length at least $2/\varepsilon$, then we distribute a (fractional) charge of $\varepsilon/2$ to each of the first $2/\varepsilon$ rectangles. 
\item Otherwise, $\gamma$ has length less than $2/\varepsilon$. Let us denote by $R'$ the last rectangle on $\gamma$. Then, 
\begin{enumerate}
\item either $R'$ is a leaf in $H$, i.e., it is not horizontally nice. In that case,  we distribute a charge of $1$ to $R'$. 
\item Or, $R'$ sees the bottom-left corner of a rectangle $R''$ that is not contained in $P_v$. Using the Claim above, we deduce that the bottom edge of $R$ is contained in a $(4/\varepsilon+1)$-fence emerging from a right vertical edge of $P_v$, i.e., $R$ is $(4/\varepsilon+3)$-protected in $P_v$. This is a contradiction with Lemma \ref{lemma:partitionningLemma}$(4)$ and Lemma \ref{lemma:partition_2pluseps_apx}$(iii)$. Thus, this case cannot happen. 
\end{enumerate}
\end{enumerate}
See Figure \ref{fig:2approx} for an illustration of the charging process. We now discuss the correctness of our charging process. First, since each node in $H$ have in-degree at most one, it is clear that a rectangle cannot be charged by two rectangles intersected by the same vertical line segment $\ell_v$. Then, using Claim \ref{claim:distance_fence},  we know that after being charged, each rectangle on $\gamma$ will be $(4/\varepsilon +3)$-protected in $P_{v'}$ for any node $v'\in T$ such that $P_{v'}\subset P_v$, and therefore rectangles that are charged will be part of the final solution $\R'$ (Lemma \ref{lemma:partition_2pluseps_apx}$(iii)$). 
}

\begin{figure}[h]
\centering
 \includegraphics[page=33,trim=30 45 45 70, clip, width=1\textwidth]{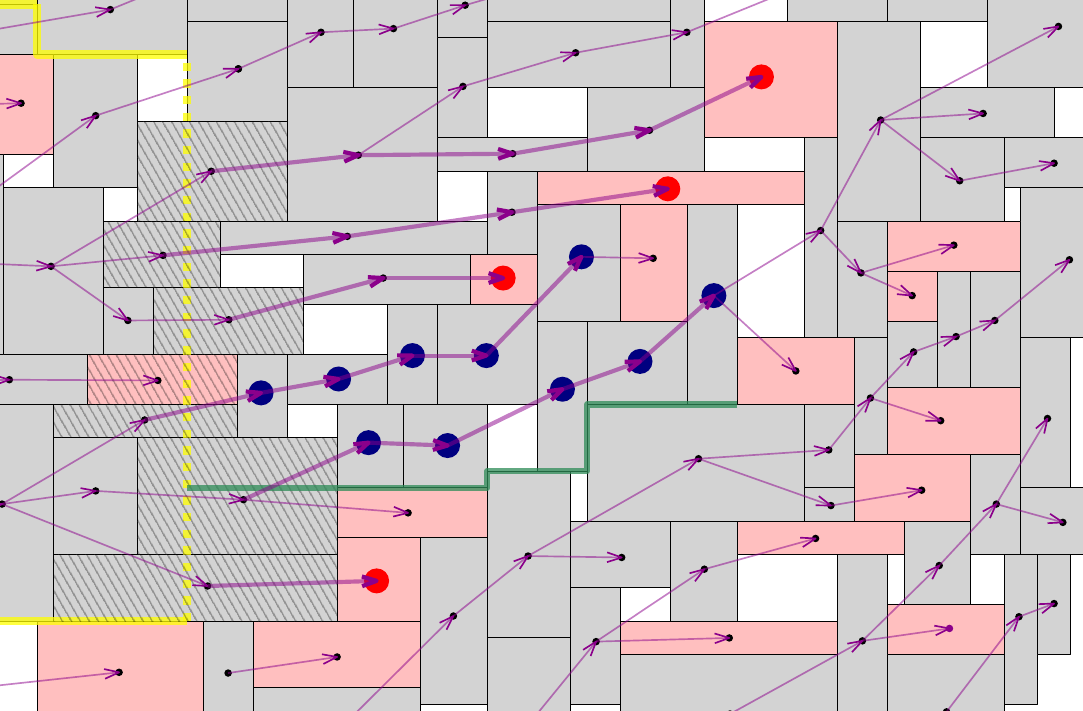}
\caption{An example of the charging process for the $(2+\varepsilon)$-approximation algorithm with $\varepsilon=\frac{2}{5}$. Rectangles that are not horizontally nice are displayed in red. For each rectangle intersected (tilling pattern) by the vertical line segment (dashed yellow line) of a cut (yellow) of a polygon $P$ during the partitioning process, we will charge some rectangles in $\OPT'(P)$ that will be protected by long enough fences (green line segments) in further iterations of recursive partitioning process. 
To decide which rectangles to charge, we use the directed forest $H$ displayed by purple arcs, and select one directed path (thicker lines) rooted at each rectangle that we need to pay for. When this path is long enough, we distribute a charge (big blue dots) of $\varepsilon/2={1}/{5}$ to each of the $2/\varepsilon=5$ first horizontally nice rectangles. Otherwise, the path ends on a non-horizontally nice rectangle and then we give a full charge of $1$ to this rectangle. 
} 
\label{fig:2approx}
\end{figure}

\paragraph{Analysis of the charging scheme.}
Now we finally show that $|\R'|\geq |\OPT'|/(2+\varepsilon)$.

\begin{lemma}
	It holds that $|\R'|\geq |\OPT'|/(2+\varepsilon)$.
\end{lemma}
\begin{proof}
	Let $b_{c},b_{\ell},r_{c},r_{\ell}$ be the fraction of blue rectangles that are charged, blue rectangles that are intersected(and hence lost), red rectangles that are charged, and red rectangles that are intersected respectively.  
	Then we have
	\begin{align}
		b_{c}+b_{\ell}+r_{c}+r_{\ell}&=1 \label{eq2:3}\\
		r_{c}+r_{\ell} &\leq \frac{1}{2} \label{eq2:2}
	\end{align}
	Here Equation \eqref{eq2:2} follows from the assumption that at least half of the rectangles are horizontally nice. 
	
	Now from the charging scheme, each saved red rectangle will receive a charge of at most $1$ and each saved blue rectangle will receive a charge of at most $\varepsilon/2$. Thus we observe that $b_\ell |\OPT'| \le  (b_c |\OPT'| \cdot \eps/2) + r_c |\OPT'|$. Hence,
	\begin{align}
		\frac{2}{\varepsilon}(b_{\ell}-r_{c}) &\leq b_{c} \label{eq2:1}
	\end{align}
	This property of charging is crucial in our analysis. 
	
	Multiplying Equation \eqref{eq2:1} by $\varepsilon$, and Equation  \eqref{eq2:2} by  $2$ and then adding them gives 
	\begin{align}
		2(b_{\ell}+r_{\ell}) &\leq \varepsilon b_{c}+1.
	\end{align}
	
	Using Equation \eqref{eq2:3}, as $b_\ell+r_\ell=1-b_c+r_c$, we obtain
	\begin{align}
		1\leq (2+\varepsilon)b_{c}+2(r_{c})  \leq (2+\varepsilon)(b_c+r_c).
	\end{align}
	This proves that $b_{c}+r_{c}\geq {1}/({2+\varepsilon})$, i.e., $|\R'|\geq |\OPT'|/(2+\varepsilon)$.
\end{proof}

\section{Omitted proofs}

\subsection{Proof of Lemma~\ref{lemma:DP-running-time}}

Polygons associated to the cells of the dynamic program are axis-parallel,
have at most $k$ edges, and each vertex has integral coordinates
in $[0,\dots,2n-1]\times[0,\dots,2n-1]$. Thus, the number of cells
is $O(n^{k})$. For each cell, the algorithm enumerates all possible
partitions of the corresponding polygon $P\in\P(k)$ into two or three
polygons $P_{1},P_{2},P_{3}$ in $\P(k)$. Such a partition can be
characterized by two sequences $\sigma$ and $\sigma'$ of line segments
whose union corresponds to edges of $P_{i}$ (with $i=1,2,3$) that
are contained in the interior of $P$. Since each $P_{i}$ has at
most $k$ edges and each such edge is shared by two polygons $P_{i}$,
the total number of line segments in $\sigma\cup\sigma'$ is at most
$3k/2$. Thus, the running time to process each DP cell is $O(n^{3k/2})$.
This gives an overall running time of $O(n^{5k/2})$. \qed

\subsection{Proof of Lemma~\ref{lemma:k-recursive-sufficient}}

Notice that we can assume w.l.o.g. that there is no pair of rectangles $R_{p},R_{q}\in\R$
such that $R_{p}\subseteq R_{q}$ (if that is the case, we can safely
replace $R_{q}$ by $R_{p}$ in any feasible solution containing $R_{q}$).

We will prove the statement by induction on the height of the tree.
If the tree has height $1$, it means that $|\R'|$ is constant and
each node aside from the root contains either a unique or no rectangle
from $\R'$. For each polygon in one of the leaves we can define a
feasible partition into a constant number of polygons from $\P$ by
circumventing the corresponding rectangle from $\R'$ and extending
its horizontal boundary until touching the boundary of the polygon.
This sequence of polygons represents a feasible partition that our
Dynamic Program considers, meaning that the returned solution will
have size at least $|\R'|$ as the polygons containing the rectangles
from $\R'$ will be leaves containing a unique rectangle.

Following the inductive proof, we can \ari{recurse} the argument on each
subtree rooted at a child of the root (notice that the subtrees are not $k$-recursive partitions of their corresponding sets of rectangles as the polygons associated to their roots are not $S$; however, that requirement does not play a role in this argumentation so, for the sake of this proof, we can assume that partitions allow to associate arbitrary polygons from $\P$ to the root node). Again the subdivision into
polygons induced by the recursive partition for $\R'$ is a feasible
candidate that our Dynamic Program considers, leading to a returned
solution of size at least $|\R'|$.\qed

\subsection{Proof of Lemma~\ref{lemma:OPT-prime-ok}}

Let $\mathcal{T}$ be a $k$-recursive partition for a set $\R'\subseteq\OPT'$,
and consider $\tilde{\R}\subseteq\OPT$ to be the set of original
rectangles that were maximally expanded so as to obtain $\R'$. In
particular $|\R'|=|\tilde{\R}|$. Since each rectangle $R\in\tilde{\R}$
is contained in a unique rectangle $R'\in{\R'}$, we conclude that
$\mathcal{T}$ is also a $k$-recursive partition for $\tilde{\R}$.
\qed

\subsection{Proof of Proposition~\ref{prop:not-both-nested}}

For a contradiction, assume that a rectangle $R\in\OPT'$ is vertically
\emph{and} horizontally nested. W.l.o.g. we can assume that its top
edge is contained in the interior of an edge of a rectangle $R'\in\OPT'$
and its left edge is contained in the interior of an edge of a rectangle
$R''\in\OPT'$. Let $c=(x,y)$ be the coordinates of the top-left
corner of $R$. Since $c$ lies in the interior of the bottom edge
of $R'$ and coordinates of rectangles in $\OPT'$ are integral, we
deduce that the point $c'=(x-1/2,y+1/2)$ must lie in the interior
of $R'$. An analogous argumentation indicates that $c'$ is also
contained in the interior of $R''$. This is a contradiction with
the fact that rectangles in $\OPT'$ are disjoint. \qed

\subsection{Proof of Lemma~\ref{lemma:charging_process}}

Let $R$ be a rectangle in $\OPT'(P)$ that is not horizontally nested.
We prove that $R$ sees a corner on its right, or is protected. The
left case is symmetrical. Consider the horizontal line segment
$h$ that connects its top-right corner to a right vertical edge $e'$
of $P$. If $h$ does not intersect any rectangle in $\OPT'(P)$,
then $R$ must be protected by a fence that emerges from a point in
$e'$, i.e., $R$ is protected in $P$. Otherwise, let $R'\in\OPT'(P)$
denote the rectangle intersected by $h$, or that has its top side
intersected by $h$, that is the closest to $R$. We now use the fact
that $R$ is not horizontally nested to prove that $R$ sees a corner
of $R'$.

\begin{figure}[H]
	\centering \includegraphics[page=12, width=1\textwidth]{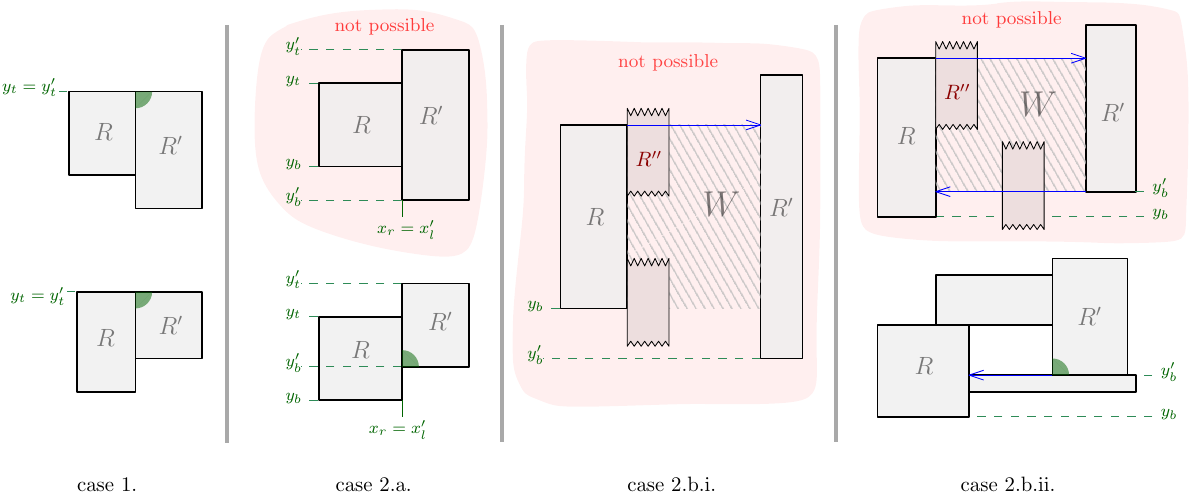}
	\label{fig:find_corner}
\end{figure}

Let $y_{b}$ and $y_{t}$ denote the $y$-coordinates of the bottom
and top edge of $R$ and let $y'_{b}$ and $y'_{t}$ denote the $y$-coordinates
of the bottom and top edge of $R'$. Additionally, let $x_{r}$ denote
the $x$-coordinate of the right edge of $R$ and let $x'_{l}$ denote
the $x$-coordinate of the left edge of $R'$. We know that $x_{r}\le x'_{l}$.
We now treat all the possible cases.

\begin{enumerate}
	\item[1.] If $y_{t}=y'_{t}$ then, by definition of $R'$, we obtain directly
	that $R$ sees the top-left corner of $R'$\footnote{In this case, one can prove that we must also have $x_{r}=x'_{l}$.}. 
	\begin{enumerate}
		\item[2.] Otherwise, we have $y'_{b}<y_{t}<y'_{t}$. 
		\begin{enumerate}
			\item[a.] If $x_{r}=x'_{l}$ then necessarily we have that $y_{b}\le y'_{b}$
			since we assumed that $R$ was not horizontally nested. In this case,
			it is clear that $R$ sees the bottom-left corner of $R'$. 
			\item[b.] Otherwise, we have $x_{r}<x'_{l}$, and then either $y'_{b}<y_{b}$
			or $y_{b}\le y'_{b}$. 
			\begin{enumerate}
				\item[i.] We show that $y'_{b}<y_{b}$ cannot happen. Indeed, assume that $y'_{b}<y_{b}$
				and consider the rectangular area $W$ with bottom-left corner $(x_{r},y_{b})$
				and top-right corner $(x'_{l},y_{t})$. In particular, the right edge
				of $R$ is the left edge of $W$. Since rectangles in $\OPT'$ are
				maximally large in each dimension, $W$ must intersect some rectangle
				in $\OPT'$. By maximality again, there must be a rectangle $R''\in\OPT'$
				that intersects $h$ within $W$, or that has its top side contained
				in $h$ within $W$. If $R''$ is contained in $P$, then we get a
				contradiction with the definition of $R'$. Otherwise, if $R''$ is
				not contained in $P$, then either $R''$ is intersected by the boundary
				of $P$, and by Lemma \ref{lemma:partitionningLemma6}(iii), it must be intersected
				by a \emph{vertical} edge of $P$, or $R''$ must be contained in
				the complement of $P$. In both cases, the horizontal rightwards ray
				from the top-right corner of $R$ must reach the boundary of $P$
				before reaching the left edge of $R'$. This is a contradiction with
				the definition of $R'$. 
				\item[ii.] If $y_{b}\le y'_{b}$ then we prove that $R$ sees the bottom-left
				corner of $R'$. As in the previous case, consider the rectangular
				area $W$ with bottom-left corner $(x_{r},y'_{b})$ and top-right
				corner $(x'_{l},y_{t})$. If $R$ does not see the bottom-left corner
				of $R'$, then there must be a rectangle in $\OPT'$ that intersects
				$W$. By maximality, this implies that there exists a rectangle $R''$
				that intersects $h$ or that has its top side contained in $h$. By
				the same argumentation as before we obtain a contradiction with the
				definition of $R'$.  \qed
			\end{enumerate}
		\end{enumerate}
	\end{enumerate}
\end{enumerate}
\subsection{Proof of Observation~\ref{obs:lemma:charging_process_extension}}
	In case $1$, we know from the proof of Lemma~\ref{lemma:charging_process} that if $R$ is non-horizontally nested, then $R$ sees the bottom left corner of $R'$ or $y_{t}=y_{t}',y_{b}'<y_{b},x_{r}=x_{l}'$ (basically the top picture in case $1$ of Figure \ref{fig:find_corner}). When $R$ is horizontally nested, we already have that $y_{t}<y_{t}', y_{b}' < y_{b} ,x_{r}=x_{l}'$ by the definition of horizontally nestedness. Which completes the proof for this case of the observation. In case $2$, if $R$ is non-horizontally nested, by the exact same arguments as in case $(2.b.i)$ in the proof of Lemma~\ref{lemma:charging_process}(refer to the figure of case $(2.b.i)$ of Figure \ref{fig:find_corner}) , we get a contradiction. Otherwise, the only case for $R$ to be horizontally nested in case $2$ of our observation is if the right edge of $R$ has to be contained in the right edge of $S$. This completes the proof of the observation

\subsection{Proof of Lemma~\ref{lemma:charged_implies_protected}}

Suppose that $R'$ receives a charge from a rectangle $R$ during
the partitioning of $P_{v}$, for some $v\in T$. In particular, $R$
sees a corner $c$ of $R'$. W.l.o.g. we assume that $R'$ is on the
right of $R$, and $c$ is the top-left corner of $R'$ (the case
where $c$ is the bottom-left corner is analogous). Consider $v'\in T$
to be the child of $v$ such that $P_{v'}$ contains $R'$. We show
that $R'$ is protected in $P_{v'}$. As discussed before, as a consequence of property (4) of Lemma \ref{lemma:partitionningLemma6},
this will imply that $R'\in\R'$. The vertical line segment $\ell_{v}$
that intersected $R$ is a vertical edge of $P_{v'}$. Since $R$
sees $c$, the horizontal leftwards ray from $c$ does not intersect
any rectangle in $\OPT'(P_{v'})$ and reaches a left vertical edge
of $P_{v'}$, say on a point $p$. Then, the horizontal edge of $R'$
that contains $c$ is contained in the line fence of $F(P_{v'})$
emerging from $p$, i.e., $R'$ is protected in $P_{v'}$. \qed

\subsection{Proof of Lemma~\ref{lemma:corner_charged_once}}

Let $c$ be a corner of a rectangle $R'\in\R'$. W.l.o.g. we assume
that $c$ lies on the left side of $R'$. We also assume that $c$
is the top-left corner of $R'$ (the case when $c$ is the bottom-left
corner is analogous). For a contradiction, assume that $c$ is charged
at least twice. Then, there are two rectangles $R_{1}\in\OPT'(P_{v_{1}})$
and $R_{2}\in\OPT'(P_{v_{2}})$, for some nodes $v_{1},v_{2}\in T$,
such that $c$ receives charges from $R_{1}$ during the partitioning
of $P_{v_{1}}$ and also receives charge from $R_{2}$ during the
partitioning of $P_{v_{2}}$. In particular, both $R_{1}$ and $R_{2}$
see $c$, and $R'\subseteq P_{v_{2}}\cap P_{v_{1}}$. We have $P_{v_{2}}\subseteq P_{v_{1}}$
or $P_{v_{1}}\subseteq P_{v_{2}}$. W.l.o.g. we assume that $P_{v_{2}}\subseteq P_{v_{1}}$.
In particular, both $R_{1}$ and $R_{2}$ are contained in $P_{v_{1}}$.

First, if $v_{1}=v_{2}$, then $R_{1}$ and $R_{2}$ are both intersected
by the same vertical line segment $\ell_{v_{1}}$ (Lemma \ref{lemma:partitionningLemma6},
property (3)). In particular, the top edge of $R_{2}$ is below the
bottom edge of $R_{1}$, or the contrary. Since both $R_{1}$ and
$R_{2}$ see $c$, the leftwards horizontal ray $h$ from $c$ intersect
the right edges of both $R_{1}$ and $R_{2}$. These two facts together
indicate that $h$ intersect the bottom edge of $R_{1}$ and the top
edge of $R_{2}$, or the contrary. This brings a contradiction with
the third condition of Definition \ref{def:seeing}.

Otherwise, we have that $v_{2}$ is a descendant of $v_{1}$. We show
that the contradiction comes here from the fact that $R_{2}$ must
be protected in $P_{v_{2}}$. Since $c$ was charged from $R_{1}$
during the partitioning of $P_{v_{1}}$, the argumentation in the
proof of Lemma \ref{lemma:charged_implies_protected} indicates that
$c$ is contained in a line fence $f$ emerging from a right vertical
side of $P_{v_{2}}$. In particular, since $R_{2}$ is contained in
$P_{v_{2}}$ and sees $c$, it must be that the top edge (or the bottom
edge) of $R_{2}$ is contained in $f$; otherwise we would get a contradiction
with the assumption that fences in $P_{v_{2}}$ do not penetrate any
rectangle in $\OPT'(P_{v_{2}})$. This implies that $R_{2}$ must
be protected in $P_{v_{2}}$. Thus, by property (4) of Lemma \ref{lemma:partitionningLemma6},
$R_{2}$ cannot be intersected during the partitioning of $P_{v_{2}}$,
and $R'$ does not receive any charge from $R_{2}$. \qed

\subsection{Proof of Lemma~\ref{lemma:factor-6}}

We can write $\OPT'=\OPT'_{h}\cup\R'\cup L$, where $\OPT'_{h}$ are
the rectangles in $\OPT'$ that are horizontally nested and $L$ are
the rectangles in $\OPT'$ not horizontally nested that are intersected
by $\ell_{v}$, for some node $v\in T$. Notice that $L\cap\R'=\emptyset$.
Since a rectangle has $4$ corners, we know from Lemma \ref{lemma:corner_charged_once}
that each rectangle in $\R'$ receives a charge of at most $2$. Then,
by Lemma \ref{lemma:charged_implies_protected}, no corners of rectangles
in $L$ are charged. Thus, $|L|\le2|\R'|$. By our initial assumption
that at most half of the rectangles in $\OPT'$ are horizontally nested,
it holds that $|L|+|\R'|\ge|\OPT'|/2$. Thus, $|\R'|\ge|\OPT'|/6=|\OPT|/6$.
\qed



\subsection{Remaining details of proof of Lemma~\ref{lemma:partitionningLemma6}}\label{apx:details-Lemma6}
We would like to ensure that $P\setminus C$ has at least two connected
components. This is clearly true if $\ell$ or $\overline{pp'}$ contain
a point in the interior of $P$. Otherwise, $\ell$ is contained in
some edge of $P$. If $\ell$ is contained in a left vertical edge
$e$ of $P$, then $p=p'$ and in particular $p$ is contained in
the interior of the left edge of some rectangle $R\in\OPT'$. Then,
also the top-left corner of $R$ is contained in $e$, and the line
segment joining the top-left corner of $R$ with the top-right corner
of $R$ is a fence emerging from $e$ that is longer than $f$, which
contradicts the choice of $f$. If $\ell$ is contained in a right
vertical edge $e'$ of $P$, then $\overline{pp'}$ must be identical
with the top or the bottom edge of $P$ (since otherwise $\overline{pp'}$
would contain a point in the interior of $P$). Assume w.l.o.g. that
it is the top edge of $P$. In this case, $P\setminus C$ has only
one connected component; however, then $\left\lfloor s/3\right\rfloor +1=1$
which implies that $P$ has at most $4s\le8$ edges. Therefore, this
is an easy case since we can afford to increase the complexity of
$P$ substantially by our cut. \mm{For example, since in this case $p'$ must be the top-right corner of a rectangle $R\in\OPT'(P)$, 
	we can define a cut $C'$
	that consists of the horizontal line segment between $p$ and the top-left corner of $R$, the left edge of $R$, and the
	bottom edge of $R$. }
Then, $P\setminus C'$ has two \mm{(horizontally convex)} connected components,
each with at most $4s+2\le10$ edges each, and $C'$ does not intersect
any rectangle from $\OPT'(P)$.

\end{document}